\documentclass[11pt]{article}

\usepackage{times}

\usepackage{amsmath,amssymb,amsthm}
\usepackage{graphicx}
\usepackage{clrscode}
\usepackage{url}
\usepackage[all]{xy}
\usepackage{ifthen}

\newtheorem{theorem}{Theorem}[section]
\newtheorem{definition}[theorem]{Definition}
\newtheorem{lemma}[theorem]{Lemma}

\newcommand{\subheader}[1]{\noindent{\bf{}{#1}}}

\newcommand{\etal}{{\it et.\ al.}}
\newcommand{\defequal}{\ensuremath{\overset{\textup{\tiny def}}{=}}}

\newcommand{\Pred}[1]{\ensuremath{\textsf{\textup{\small{}#1}}}}
\newcommand{\tinyPred}[1]{\ensuremath{\textsf{\textup{\scriptsize{}#1}}}}

\newcommand{\ComputationSet}[1]{\ensuremath{\mathcal{C}\left({#1}\right)}}
\newcommand{\computationSet}[1]{\ensuremath{\mathcal{C}({#1})}}

\newcommand{\PCName}{partition consistency}
\newcommand{\POBName}{partial-order broadcast}
\newcommand{\NWName}{network}

\newcommand{\PCShort}{PC}
\newcommand{\POBShort}{POB}

\newcommand{\PC}[1]{\ifthenelse{\equal{#1}{}}{\ensuremath{\text{\PCShort}}}{\ensuremath{\text{\PCShort}(#1)}}}
\newcommand{\POB}[1]{\ifthenelse{\equal{#1}{}}{\ensuremath{\text{\POBShort}}}{\ensuremath{\text{\POBShort}(#1)}}}
\newcommand{\NW}{NW}

\newcommand{\tstamp}{TS}
\newcommand{\token}{TKN}

\DeclareRobustCommand{\smartBoxArg}[1]{\ifthenelse{\equal{#1}{}}{}{\ensuremath{[#1]}}}
\DeclareRobustCommand{\smartRoundArg}[1]{\ifthenelse{\equal{#1}{}}{}{\ensuremath{(#1)}}}
\DeclareRobustCommand{\PredNW}[1]{\ensuremath{\Pred{NW}\smartBoxArg{#1}}}
\DeclareRobustCommand{\PredNWVoid}{\ensuremath{\Pred{NW}}}
\DeclareRobustCommand{\PredPC}[2]{\ensuremath{\Pred{PC}\smartRoundArg{#1}\smartBoxArg{#2}}}
\DeclareRobustCommand{\PredPCVoid}{\ensuremath{\Pred{PC}}}
\DeclareRobustCommand{\PredPOB}[2]{\ensuremath{\Pred{POB}\smartRoundArg{#1}\smartBoxArg{#2}}}
\DeclareRobustCommand{\PredPOBVoid}{\ensuremath{\Pred{POB}}}

\DeclareRobustCommand{\WitOrd}[3]{\ensuremath{\Pred{Witnesses}[{#1},{#2},{#3}]}}

\DeclareRobustCommand{\impl}[4]{\ensuremath{\text{#1}^{\text{\tiny #2}}_{\text{\tiny #3}}\smartRoundArg{#4}}}

\DeclareRobustCommand{\tsimpl}[2]{\impl{\tstamp{}}{\POB{#1}}{\NW{}}{#2}}
\DeclareRobustCommand{\tokenimpl}[2]{{\impl{\token{}}{\POB{#1}}{\NW{}}{#2}}}

\DeclareRobustCommand{\topimpl}[4]{{\impl{#1}{\PC{#2}}{\POB{#3}}{#4}}}

\DeclareRobustCommand{\swimpl}[3]{\topimpl{SWFR}{#1}{#2}{#3}}
\DeclareRobustCommand{\fwimpl}[3]{\topimpl{FWSR}{#1}{#2}{#3}}
\DeclareRobustCommand{\wimpl}[3]{\topimpl{WR}{#1}{#2}{#3}}

\newcommand{\MapSet}[1]{\ensuremath{\langle\; #1 \;\rangle}}
\newcommand{\targ}[1]{\ensuremath{\widehat{#1}}}

\newcommand{\orderArrow}[1]{\ensuremath{\xrightarrow{{\scriptscriptstyle #1}}}}

\newcommand{\simArrow}[1]{\ensuremath{\underset{{#1}}{{\sim}}}}

\newcommand{\lift}[1]{lift$(\mbox{#1})$}
\newcommand{\Lift}[1]{Lift$(\mbox{#1})$}

\newcommand{\seqlinearp}[1][p]{\ensuremath{L_{#1}}}
\newcommand{\seqLinearp}[1][\widehat{p}]{\ensuremath{\widehat{L}_{#1}}}

\newcommand{\rellinearp}[1][p]{\ensuremath{\orderArrow{{\seqlinearp[{#1}]}}{}}}
\newcommand{\relLinearp}[1][\widehat{p}]{\ensuremath{\orderArrow{{\seqLinearp[{#1}]}}}}

\newcommand{\delorder}{\textnormal{delOrder}}
\newcommand{\programOrder}{\ensuremath{\textnormal{prog}}}
\newcommand{\programOrderp}{\ensuremath{\textnormal{prog}_p}}
\newcommand{\ProgramOrder}{\ensuremath{\widehat{\textnormal{prog}}}}
\newcommand{\HappensBefore}{\ensuremath{\widehat{\textnormal{HappensBefore}}}}
\newcommand{\happensBefore}{\ensuremath{\textnormal{HappensBefore}}}
\newcommand{\MessageOrder}{\ensuremath{\widehat{\textnormal{MessageOrder}}}}
\newcommand{\messageOrder}{\ensuremath{\textnormal{MessageOrder}}}
\newcommand{\WritesInto}{\ensuremath{\widehat{\textnormal{WritesInto}}}}
\newcommand{\writesInto}{\ensuremath{\textnormal{WritesInto}}}
\newcommand{\FifoChannel}{\ensuremath{\widehat{\textnormal{FifoChannel}}}}
\newcommand{\fifoChannel}{\ensuremath{\textnormal{FifoChannel}}}

\newcommand{\mfalse}{\const{false}}

\newcommand{\Processes}{\ensuremath{P}}

\newcommand{\computationC}{\ensuremath{C}}
\newcommand{\ComputationC}{\ensuremath{\widehat{C}}}

\newcommand{\cutwrites}{\ensuremath{{\tinyPred{wrts}}}}
\newcommand{\cutwritesgroup}[1]{\ensuremath{{\tinyPred{wrts}(#1)}}}
\newcommand{\cutdeliverslabel}[1]{\ensuremath{{\tinyPred{delivers}(#1)}}}
\newcommand{\cutmultiwriters}{\ensuremath{{\tinyPred{multi-wrtrs}}}}

\newcommand{\myconst}[1]{\textit{\footnotesize #1}}

\newcommand{\Set}[1]{\ensuremath{\left\{#1\right\}}}

\providecommand{\abs}[1]{\lvert#1\rvert}

\newcommand{\memCon}{\Pred{MC}}
\newcommand{\MemCon}{\ensuremath{\widehat{\Pred{MC}}}}

\newcommand{\FilterMemCon}[2]{\ensuremath{\computationSet{{#2}, {#1}}}}

\newcommand{\FiltermemConVoid}[1]{\FilterMemCon{\memCon}{#1}}
\newcommand{\FilterMemConVoid}[1]{\FilterMemCon{\MemCon}{#1}}

\newcommand{\transform}[1]{\ensuremath{\tau({#1})}}

\newcommand{\opCommon}[3][{}]{\ensuremath{\ifthenelse{\equal{#2}{low}}{{\textnormal{\texttt{\small{}#3}}}_{#1}}{{\textnormal{\textsc{\small{}#3}}}_{#1}}}}

\newcommand{\opCommonArg}[4][{}]{\ensuremath{\opCommon[#1]{#2}{#3}(#4)}}

\newcommand{\opNoReturn}[4][{}]{\ensuremath{{\opCommonArg[#1]{#2}{#3}{#4}}}}

\newcommand{\opReturn}[5][]{\ensuremath{\genfrac{}{}{}{1}{\opCommonArg[#1]{#2}{#3}{#4}}{#5}}}

\newcommand{\codeOp}[2]{\ensuremath{\func{#1}(#2)}}
\newcommand{\codeOpPattern}[3]{\ensuremath{{#3} \dashv \codeOp{#1}{#2}}}

\newcommand{\opwriteVoid}{\opCommon{}{write}{}}
\newcommand{\opreadVoid}{\opCommon{}{read}{}}
\newcommand{\opWriteVoid}{\opCommon{low}{write}{}}
\newcommand{\opReadVoid}{\opCommon{low}{read}{}}

\newcommand{\opReadp}[3]{\opReturn[#3]{low}{read}{#1}{#2}}

\newcommand{\opwrite}[1]{\opNoReturn{}{write}{#1}}
\newcommand{\opwritep}[2]{\opNoReturn[#2]{}{write}{#1}}
\newcommand{\opWrite}[1]{\opNoReturn{low}{write}{#1}}
\newcommand{\opWritep}[2]{\opNoReturn[#2]{low}{write}{#1}}

\newcommand{\opread}[2]{\opReturn{}{read}{#1}{#2}}

\newcommand{\opreadpNoReturn}[2]{\opNoReturn[#2]{}{read}{#1}}
\newcommand{\opreadNoReturn}[1]{\opNoReturn{}{read}{#1}}
\newcommand{\opRead}[2]{\opReturn{low}{read}{#1}{#2}}
\newcommand{\opReadNoReturn}[1]{\opNoReturn{low}{read}{#1}}

\newcommand{\opEnqueue}[1]{\opNoReturn{low}{priority-enQ}{#1}}
\newcommand{\opEnqueueVoid}{\opCommon{low}{priority-enQ}{}}

\newcommand{\opEnqueuep}[2]{\opNoReturn[#2]{low}{priority-enQ}{#1}}

\newcommand{\opDequeuep}[3]{\opReturn[#3]{low}{extractmin}{#1}{#2}}
\newcommand{\opDequeueVoid}{\opCommon{low}{extractmin}{}}

\newcommand{\opDequeue}[2]{\opReturn{low}{extractmin}{#1}{#2}}

\newcommand{\opDequeueNoReturn}[1]{\opNoReturn{low}{extractmin}{#1}}

\newcommand{\opPeekminNoReturn}[1]{\opNoReturn{low}{peek-min}{#1}}
\newcommand{\opIsemptyNoReturn}[1]{\opNoReturn{low}{priority-isempty}{#1}}

\newcommand{\opFIFOEnqueue}[1]{\opNoReturn{low}{fifo-enQ}{#1}}
\newcommand{\opFIFOEnqueueVoid}{\opCommon{low}{fifo-enQ}{}}

\newcommand{\opFIFOPeekHeadNoReturn}[1]{\opNoReturn{low}{peek-head}{#1}}

\newcommand{\opFIFODequeue}[2]{\opReturn{low}{fifo-deQ}{#1}{#2}}
\newcommand{\opFIFODequeueNoReturn}[1]{\opNoReturn{low}{fifo-deQ}{#1}}
\newcommand{\opFIFODequeueVoid}{\opCommon{low}{fifo-deQ}{}}

\newcommand{\opFIFOIsEmptyNoReturn}[1]{\opNoReturn{low}{isempty}{#1}}

\newcommand{\opEitherEnqueue}[1]{\opNoReturn{low}{enQ}{#1}}
\newcommand{\opEitherEnqueueVoid}{\opCommon{low}{enQ}{}}

\newcommand{\opEitherDequeueVoid}{\opCommon{low}{deQ}{}}
\newcommand{\opEitherDequeuep}[3]{\opReturn[#3]{low}{deQ}{#1}{#2}}
\newcommand{\opEitherDequeuepVoid}[1]{\opCommon[#1]{low}{deQ}{}}
\newcommand{\opEitherEnqueuep}[2]{\opNoReturn[#2]{low}{enQ}{#1}}

\newcommand{\opsend}[1]{\opNoReturn{}{send}{#1}}
\newcommand{\opsendVoid}{\opCommon{}{send}{}}
\newcommand{\opSend}[1]{\opNoReturn{low}{send}{#1}}
\newcommand{\opSendVoid}{\opCommon{low}{send}{}}

\newcommand{\codeRecvPattern}[1]{\codeOpPattern{recv}{}{(#1)}}
\newcommand{\oprecv}[1]{\opReturn{}{recv}{}{#1}}
\newcommand{\oprecvVoid}{\opCommon{}{recv}{}}

\newcommand{\opRecvNoReturn}{\opNoReturn{low}{recv}{}}
\newcommand{\opRecvVoid}{\opCommon{low}{recv}{}}
\newcommand{\opRecv}[1]{\opReturn{low}{recv}{}{#1}}

\newcommand{\opbcast}[1]{\opNoReturn{}{bcast}{#1}}
\newcommand{\opbcastp}[2]{\opNoReturn[#2]{}{bcast}{#1}}
\newcommand{\opbcastVoid}{\opCommon{}{bcast}{}}
\newcommand{\opbcastpVoid}[1]{\opCommon[#1]{}{bcast}{}}
\newcommand{\opBcast}[1]{\opNoReturn{low}{bcast}{#1}}
\newcommand{\opBcastVoid}{\opCommon{low}{bcast}{}}

\newcommand{\opdeliver}[1]{\opReturn{}{deliver}{}{#1}}
\newcommand{\opdeliverNoReturn}{\opNoReturn{}{deliver}{}}
\newcommand{\opdeliverVoid}{\opCommon{}{deliver}{}}
\newcommand{\opDeliverNoReturn}{\opNoReturn{low}{deliver}{}}
\newcommand{\opDeliverVoid}{\opCommon{low}{deliver}{}}
\newcommand{\opDeliver}[1]{\opReturn{low}{deliver}{}{#1}}
\newcommand{\opdeliverp}[2]{\opReturn[#2]{}{deliver}{}{#1}}
\newcommand{\opdeliverpVoid}[1]{\opCommon[#1]{}{deliver}{}}

\newcommand{\blankfrac}[2]{\ensuremath{\genfrac{}{}{0pt}{}{#1}{#2}}}

\begin{document}

\begin{titlepage}

\title{Partition Consistency:\\ A Case Study in Modeling Systems with Weak Memory Consistency and Proving Correctness of their Implementations}

\author{Steven Cheng, Lisa Higham, Jalal Kawash \\
   University of Calgary \\
   stevechy@gmail.com,\{higham,jkawash\}@ucalgary.ca \\
   }

\date{}

\maketitle \thispagestyle{empty}


\begin{abstract}
Multiprocess systems, including grid systems, multiprocessors and multicore computers, 
incorporate a variety of specialized hardware and software mechanisms, 
which speed computation, but result in complex memory behavior. 
As a consequence, the possible outcomes of a concurrent program can be unexpected.
A \emph{memory consistency model} is a description of the behaviour of such a system.
Abstract memory consistency models aim to capture the concrete implementations and architectures. 
Therefore, formal specification of the implementation or architecture is necessary, 
and proofs of correspondence between the abstract and the concrete models are required. 

This paper provides a case study of this process.
We specify a new  model, \emph{partition consistency}, 
that generalizes many existing consistency models.
A concrete message-passing network model is also specified. 
Implementations of partition consistency on this network model are then presented and proved correct. 
A middle level of abstraction is utilized to facilitate the proofs. 
All three levels of abstraction are specified using the same framework.
The paper aims to illustrate a general methodology and techniques 
for specifying memory consistency models and proving the correctness of their implementations. 
\end{abstract}

\bigskip

\end{titlepage}


\section{Introduction}
\label{intro.section}

Multiprocess systems including networks, grid systems, multiprocessors and multicore computers, 
incorporate a variety of specialized hardware and software mechanisms
such as replicated memory, multi-level caches, 
write-buffers, multiple buses, and complex support for message passing. 
These features help speed computation by hiding latency and avoiding bottlenecks when accessing shared memory. 
An unfortunate consequence is that the possible outcomes of concurrent programs can be unexpected.
Because the processes' views of the current state of memory at any moment do not
completely agree, the execution of a concurrent program may not be sequentially consistent.

Sequential consistency is important. It can be efficiently implemented in the absence of data-races; 
it supports platform independent code; it is easier to reason about sequential consistency than weaker models. 
Nonetheless, real systems typically deviate from implementing sequential consistency, resulting in complex memory behavior.

It is therefore necessary to formally and precisely specify what computations can arise on a given system.
Such a specification is a \emph{memory consistency model}.
A concrete or low-level memory consistency model would describe possible outcomes in terms of the behavior of the actual 
system due to its hardware and software architecture. 
For example, consider a system where each process has a write-buffer. 
The behaviour of each store($x,\nu$) instruction could be separated into a sequence of low-level events: 
first the value $\nu$ for location $x$ is recorded in the local write-buffer,
later, the pair $(x, \nu)$ is removed from the write-buffer and the shared memory location $x$ is updated to contain $\nu$. 
Similarly, the behaviour of each load($x$) instruction could be separated into the events: 
consult the write-buffer for a pending store to location $x$; 
if one exists, return the associated value; 
otherwise fetch the value of location $x$ from main memory and return that value.
Such an operational specification in term of the events that occur in the system when an instruction is executed 
is a good way to capture exactly what can arise when a concurrent program is executed.
A programmer, however, should be spared having to deal with the architectural details;
she should be able to reason about her program in terms of the instructions of the program, rather than low level events.
An abstract memory consistency model aims to provide such a non-operational specification of the behaviors of concrete architectures and systems.
How can we be sure that such an abstract model is correct? 
What is required is a specification of both the abstract model and the concrete architecture and a proof of their equivalence. 

This paper provides a case study of this process.
We introduce a new general class of memory consistency models, called \emph{\PCName{}},
which captures different degrees of consistency between various sets of shared variables, 
as is common in distributed and multiprocessor environments.  
Sequential consistency \cite{lam79}, Goodman's processor consistency \cite{aha93}, and the pipelined random-access machine \cite{lip88} are all instances of \PCName{}.
The definition of \PCName{} is a natural extension of sequential consistency. 
Implementations of \PCName\ on the message-passing network model are presented and proved correct.
In contrast to the high-level \PCName{} definition, the message passing network 
is modelled at the level of message sends and receives combined with operations on local memory.
It requires several partial orders to capture the relationships between message operations and local memory operations.
Although \PCName\ is an over-simplification, it is non-trivial and it captures some essential properties of  actual implementations. 



Our implementations and their proofs of correctness are facilitated by 
a middle level of abstraction, 
which generalizes a totally ordered broadcast model to a partially ordered one.
Thus we introduce three levels of abstraction:  
the abstract \PCName{} model, 
the intermediate \POBName{} model, and 
the concrete message-passing network model. 
The same framework is used to define the memory consistency model of each level. 
So the intermediate \POBName{} model is first the 
\emph{target} of the implementation of the 
\emph{specified} \PCName{} model. 
Next, the same \POBName{} abstraction serves as the
specified model, which is implemented on the target message-passing network model.
We define the fast-read/slow-write and fast-write/slow-read
implementations of \PCName{} on the \POBName{} model. 
Next we give two implementations of \POBName{} (one token-based and one timestamp-based) on message-passing networks. 
This results in four compositions of transformations; 
each implements \PCName{} on message-passing networks.

Proofs of correctness of implementations of shared memory models on multiprocessors or networks 
typically involve a great deal of tedious but essential detail, 
and are thus prone to imprecision and error. 
Layering the implementation on levels that are defined using a common framework helps to overcome these problems
by allowing us to focus on only part of the proof obligation at each level. 
We also introduce a diagrammatic notation that provides
a visual representation of logical statements and inferences.
The  precision and conciseness of mathematical logic aids in avoiding some potential ambiguities of consistency proofs, 
while its representation in diagram format helps keep our intuitions aligned with the proof. 
This notation helped us uncover errors in our initial implementations and our attempted proofs. 

Partition consistency is of independent interest because it supports different consistency guarantees for different partitions of variables.
Itanium \cite{IntelArch2} and Java \cite{DBLP:conf/popl/MansonPA05}, for example,
exhibit differing degrees of consistency depending on how variables are declared. 
Similarly, abstract consistency models are typically defined by partitioning operations into classes, 
where processes have different levels of agreement on the order of operations in each class.
For example,
Sequential consistency (SC) \cite{lam79} requires complete agreement on all operations; 
pipelined random-access machine (P-RAM) \cite{lip88} requires agreement on the write operations by any individual process; 
in addition to P-RAM, processor consistency (PC-G) \cite{aha93} requires agreement on the operations on the same variables. 
SC, P-RAM, and PC-G are all special cases of \PCName{}. 
This paper provides implementations of any instance of the \PCName{} class 
on message-passing networks with multi-threaded nodes. 

Some of the authors have used earlier versions of this framework in previous research. For example, it was used to expose problems with the Java specifications when applied to long-lived programs \cite{HK98}, to provide a simple abstract definition for TSO \cite{DBLP:journals/tocs/HighamJK07}, to study the Intel Itanium memory consistency model \cite{DBLP:journals/entcs/HighamJK07}, and to compare Itanium with the SPARC models \cite{DBLP:conf/spaa/HighamJ06}.

\subsubsection*{Organization of the rest of the paper:} 
Section \ref{related-work.section} discusses related research on modeling and proof techniques as applied to weak memory consistency models.
Section \ref{model.section} provides the framework and the three levels of abstraction used in this paper, 
by defining the memory consistency of each level.  
The setup for transformations and their proofs of correctness is given in Section \ref{setup.section}.
The implementations of \PCName{} on \POBName{} and their proofs of correctness are given in Section \ref{implement.section}. 
Section \ref{token.section} (respectively, \ref{timestamp.section}) provides 
the token (respectively, timestamp) implementation of \POBName{} on the message-passing network model,
and the proof that the implementation is correct.
Section \ref{concl.section} concludes by summarizing the paper and discussing future directions.

%

\section{Related Work}
\label{related-work.section}

Research concerning systems with weak memory consistency proceeds in several directions.
In the following we look at those directions that are related to our work. 
We first review some formalisms used to specify given systems focusing on the ones that closely resemble our framework. 
There is a proliferation of abstract memory consistency models aimed at capturing those relaxations of 
sequential consistency that are common in real systems.
We concentrate the next part of our review on the subset of these that are instances of partition consistency.
Our work implements partition consistency on the message-passing network model, 
so we next discuss its relationship to the large body of research concerned with
how to make algorithms that are correct for sequentially consistent machines 
remain correct when run on systems with weaker memory consistency guarantees.
The principal goal of this paper is to demonstrate (through the case study of partition consistency) 
one strategy for proving the  correctness of an implementation of an abstract memory consistency model on a concrete operational model.
While an abstract memory consistency model is often used to provide a non-operational model of 
a multiprocess system, 
a proof of the correctness of the proposed model is less common.
Our final review section discusses other instances of this activity, and related proof techniques.

\subsection{Formalisms for modelling systems}

There are several formalisms use to specifying concurrent systems --- 
the most common types are based on process algebras and automata theory.
Process algebras start with a basic set of processes, then combine them into larger systems using various algebraic operators.
Communicating sequential processes (CSP) \cite{HoareCSP1985} and 
the calculus of communicating systems (CCS) \cite{Milner199509} are classic examples.

In the input-output automata (IOA) language \cite{Lynch199704}, processes are specified as automata that
communicate by action synchronization. Actions are designated as input, output, or internal, where 
each input action is jointly performed with all matching output actions.
Actions in IOA are automata transitions that can arbitrarily modify local state.
Their effect on local state is specified in an imperative language that is essentially pseudocode.
This can make IOA useful for reasoning about algorithms written in imperative programming languages.
The input/output matching is similar to CCS,
 and the fact that more than two processes can participate in a communication action is similar to CSP.

The temporal logic of actions, TLA \cite{Lamport200207}, specifies automata using a mathematical language.
It avoids programming languages, making it attractive to hardware designers.
TLA also provides many tools for reasoning about automata in general, and can be used with IOA
 \cite{DBLP:conf/tphol/Muller98}.

Another approach is to regard a system as a collection of processes, 
where each process produces a sequence of instructions invocations (program order).
A memory consistency model is a set of partial order constraints on the collection of all instructions of the system.
(The model will specify what subsets of program order must be maintained by which processes (their ``views''), 
and to what degree processes' views must agree.)
This is the approach used in this paper;  it is defined in Section \ref{model.section}.
It is most similar to that used by Steinke and Nutt \cite{DBLP:journals/jacm/SteinkeN04},
who observed that all of the models that they were aware of from the literature could be expressed 
in terms of the existence of serial views that extend certain partial orders.

\subsection{Memory consistency models}

There is a large body of literature examining various memory consistency models.
Steinke and Nutt \cite{DBLP:journals/jacm/SteinkeN04}
re-specified several models including Goodman's processor consistency, sequential consistency, pipelined random-access, and causal consistency 
in terms of processes' views, and 
proved that their definitions are equivalent to those in the literature. 
Then, they used their partial order definitions to compare the relative strengths of  models.
Specifically, they arranged 12 models into a lattice, with SC being the strongest model.  
This research demonstrated a definitional style that is general enough to capture 
the models in the literature and to facilitate their comparisons.
We now informally describe some specific memory consistency models. 
It is straightforward to recast each using the partial order formalism of Steinke and Nutt as 
adapted in Section \ref{model.section}.

Sequential consistency is a strong memory consistency model introduced by Lamport \cite{lam79}; 
it requires agreement between all system processes on a single view of the all the operations of all processes. 
This view agrees with the order in which the instructions producing these operations appear in their programs, 
called \emph{program order}.
An even stronger model, \emph{atomic objects} as defined by Lamport and by Lynch \cite{Lyn96}, 
and \emph{linearizability}, defined by Herlihy and Wing \cite{her90}
requires agreement on global timing of operations in addition to sequential consistency.  

Lipton and Sandberg \cite{lip88} introduced a much weaker consistency model than sequential consistency, 
the pipelined random-access machine (P-RAM) memory model. P-RAM is an example of distributed-shared memory (DSM).
It requires a process's view to include its own operations and all other processes' writes. 
The view of a process must be consistent with program order. 
However, P-RAM allows processes to disagree on the order of two writes performed by two different processes. 
As a result, P-RAM is so weak that it cannot support a solution to mutual exclusion 
with only read/write variables \cite{DBLP:journals/tpds/HighamK06}.

Coherence requires that processes agree on the ordering of operations on each object separately, 
but not on how the operations on different objects interleave. 
Coherence is also too weak to support mutual exclusion with only read/write variables \cite{HK97}. 
Coherence is a memory consistency model that captures a property that we might expect from any multiprocessing system. 
Nevertheless, some language memory models, such as Java, are incomparable to coherence.  

Goodman's version of processor consistency (PC-G) \cite{goo89}, 
as formalized by Ahamad \etal{} \cite{aha93}, strengthens P-RAM by adding coherence to it. 
PC-G executions must simultaneously satisfy both P-RAM and coherence. 
PC-G is weaker than SC, 
but it supports mutual exclusion with only read/write variables \cite{DBLP:journals/tpds/HighamK06}. Hence, 
PC-G is one of the few weak models that can be used to implement SC 
with only read/write variables \cite{HighamKawash2008}.
Other versions of processor consistency exist, and these versions are incomparable \cite{verthesis}.

In this paper, we introduce \emph{partition consistency}, 
which defines a family of memory consistency models inspired by PC-G, P-RAM, and SC. 
Each of these three models is a special case of partition consistency.

In addition to such abstract memory consistency models, 
the literature contains formalizations of the memory consistency models implemented by concrete multiprocessor machines.
Higham, Jackson, and Kawash explore memory consistency models for 
SPARC \cite{DBLP:journals/tocs/HighamJK07} and Itanium \cite{HJK-disc-2006,HJK-ICDCN-2006} multiprocessors, 
including the TSO model, 
which is claimed to be the consistency model of Intel x86 multiprocessors. 
The consistency model for Alpha processors is described in the Alpha manual \cite{CCC98} 
and is formally defined and investigated by Attiya and Friedman \cite{af94}.
The PowerPC consistency model is formalized by Corella, Stone, and Barton \cite{CSB94}.
Sarkar, Sewell, Alglave, Maranget and Williams \cite{DBLP:conf/pldi/SarkarSAMW11} also 
aim at faithfully representing the memory model of POWER multiprocessors. 
This research defines an abstract machine that implements the model, 
and provides programmers with a high level explanation of how the memory model is implemented by POWER multiprocessors.
A large number of manually coded and automatically generated litmus tests are
run on various POWER processors to provide confidence that the hardware behaves as the memory model predicts.
The Intel architecture developer manual provides some description of an x86 memory model \cite{intel64softwaredevmanual},
which is also clarified in an Intel whitepaper \cite{intel64memwhitepaper}. 
Owens, Sarkar, Sewell, Nardelli, Ridge, Baribant, Myreen, and Alglave 
studied the memory consistency model of the x86 architecture 
\cite{DBLP:journals/cacm/SewellSONM10,DBLP:conf/popl/SarkarSNORBMA09, DBLP:conf/tphol/OwensSS09}.

It is important to study the behavior of highly used programming languages when there is more than one thread of execution and to formalize the resulting memory models.
 The Java memory model was the subject of a few studies (for instance, see \cite{DBLP:conf/popl/MansonPA05}, \cite{DBLP:conf/tphol/AspinallS07}, and \cite{HK98}). The C++ memory model was studied by Boehm and Adve \cite{Boehm:2008:FCC:1379022.1375591} and recently formalized by Batty, Owens, Sarkar, Sewell, and Weber \cite{DBLP:conf/popl/BattyOSSW11}. The partition consistency model is simpler than the models arising from such languages; our focus in this paper is on proving the correctness of implementations of a model.

\subsection{Implementations of sequential consistency on systems with weaker memory consistency guarantees}
\label{dsm-related-work.subsec}

Many researches have address the question of how to ensure that a program that is correct under sequential consistency 
remains correct and efficient under a weaker consistency model.

Attiya and Welch provide DSM implementations for sequential consistency and linearizability\cite{DBLP:journals/tocs/AttiyaW94}, 
and they examine the difference in these implementations in terms of message delay.
Building on Lipton and Sandberg's lower bound on sequential consistency \cite{lip88}, 
Attiya and Welch establish that it is more expensive to implement linearizability 
than it is to implement SC.
Cholvi, Fernandez, Jimenez, and Raynal \cite{DBLP:conf/nca/CholviFJR04} improve the best-case performance given by Attiya and Welch 
by showing that sometimes a read can be guaranteed not to incur any message delay. 
Cholvi \etal{} present a sequentially consistent DSM protocol that ensures fast writes, 
but not all reads can be fast. 
Their implementation uses a single circulated token to synchronize the processors' copies of memory.

In this paper, we generalize Attiya and Welch's total-order broadcast algorithm \cite{DBLP:journals/tocs/AttiyaW94} 
to a partial-order one, 
allowing us to a implement a class of weak memory consistency models. 
Since our focus is weak memory consistency, our generalization is based on their implementation for SC, 
rather than  linearizability.
Brzezinski and Szychowiak \cite{DBLP:conf/iscis/BrzezinskiS03} provide a DSM implementation for PC-G and prove its correctness.
It statically assigns a master node to each variable to ensure coherence.
In contrast, our implementations use a timestamp protocol or circulating tokens and is fully distributed. 

Agrawal, Choy, Leong and Singh created the Maya DSM \cite{DBLP:conf/hpdc/AgrawalCLS94} 
to experiment with weak memory consistency models.
Amza, Cox, Dwarkadas, Keleher, Lu, Rajamony, Yu, and Zwaenepoel implemented the weak memory consistency model 
release consistency in their Treadmarks DSM \cite{DBLP:conf/usenix/KeleherCDZ94}. 
Adve introduced data-race-free (DRF) programs \cite{adv96,DBLP:journals/cacm/Adve10}.
DRF programs have that property that there are no data races in any sequentially consistent execution, 
which can be achieved by insertion of appropriate synchronization instructions.
DRF programs are guaranteed to yield sequentially consistent computations on several weak models including 
release consistency \cite{gha90,gha93}.

Shasta and Snir \cite{DBLP:journals/toplas/ShashaS88} implement sequential consistency on MIMD machines such as the NYU Ultracomputer and IBM RP3.
In such machines, 
a packet-switched network connects processors to multi-ported memory modules that can be simultaneously accessed.
Their goal is to gain efficiency by exploiting potential simultaneous accesses to memory.
To ensure that sequential consistency is not violated, 
control instructions are added to delay accesses until the previous one by the same processor is completed, and 
synchronization code (locks) are used to deal with cases when some memory accesses need to have stronger atomicity than the word-level atomicity provided by the machine.  
For efficiency,  it is important to the use these constructs only when necessary.
Analysis of interdependence of processes is used to minimize their use.
By doing this analysis of a program before it is executed, 
they show that their implementation ``requires far less locking than database control theory would lead one to expect''.
Their proofs are primarily set-theoretic in structure. In this paper, we similarly implement sequentially consistency (and other models) but on a message-passing platform.

Kuperstein, Vechev and Yahav \cite{DBLP:conf/fmcad/KupersteinVY10} developed an algorithm and its implementation (Fender) 
that infers where memory fences are needed to maintain correctness.
Fender takes as input a finite state program, a safety specification and a memory model described by a 
transition system. 
For each state, Fender computes an avoid formula that captures all the ways to prevent an execution from reaching the state.
Once transitions to invalid states are identified, provided they can be avoided by local fences, 
such fences can be inserted to ensure that the invalid states are not reachable.
This approach is distinctly different from the approach of this paper.  
It uses an operational definition of the memory consistency model, and a state-based notion of safety,
whereas we use a partial-order definition of computations and a predicate on computations to 
define safety.
Our techniques, however, do not provide an automated way to infer where and what kind of synchronization is required.

The CheckFence tool of Burckhardt, Alur, and Martin \cite{DBLP:conf/pldi/BurckhardtAM07} is another tool to ensure that
programs remain correct when executed under weak memory models. 
It takes as input a program written in a subset of C and an axiomatic memory model, and 
determine if there is an execution that violates its specification. 
CheckFence works by combining the axiomatic memory model definition with a compiled version of input program to 
to form a boolean satisfiability problem.
It then calls a SAT solver to find violating executions for finite unrolling of the program.
Like Fender, the advantage of CheckFence is that much of the work of checking correctness is automated.
It does not, however, directly shed much light on how to fix a program that does permit executions that violate the specification.

Alglave and Maranget \cite{DBLP:conf/pldi/SarkarSAMW11} provide a class of memory models that can be instantiated to produce sequential consistency, 
Sparc-TSO and a model based on Power processors.
They provide theorems on barrier placement needed to regain sequential consistency 
and a tool called \emph{diy} to automatically generate litmus tests that detect relaxations of SC.
They use a specification style that is similar to what we use in this paper.

Huseynov's Distributed Shared Memory webpage \cite{dsmweb} tracks the available academic and commercial DSM implementations.

In contrast with all of these papers, 
this paper is concerned with developing techniques to prove that a given (weak) abstract memory consistency model is a correct abstraction of a given architecture.

\subsection{Proofs of correctness of concurrent systems}

There are several methods to prove the correctness of programs.
Hoare established a logic-based approach, which proceeds by showing that a program satisfies specified post-conditions given that its satisfies specified pre-conditions (see Backhouse \cite{Backhouse03}). 
Owicki and Gries extended Hoare's logic to apply to  multiprocessor systems (see Feijen and van Gasteren \cite{FeijenVanGasteren99}). 
Reynolds generalized Hoare's logic to separation logic \cite{DBLP:conf/lics/Reynolds02}, 
which facilitates proofs of programs because it allows reasoning about parts of memory independent of the entire global state. 
Concurrent separation logic \cite{DBLP:journals/tcs/Brookes07} combines these two extensions. 
It is aimed at reasoning about concurrent programs and their shared mutable data structures.

CCS \cite{Milner199509} and $\pi$-calculus \cite{AcetoReactiveSystems2007} establish correctness proofs by simulating one automaton with another. 
Proofs proceed by showing that the properties of a simulation imply the specifications. 
In these approaches, a system execution is a sequence of events. 
Lamport's temporal logic of actions (TLA) \cite{Lamport200207} works similarly except that an execution in TLA is a sequence of states.
Lamport's system executions framework \cite{lam86b} uses a collection of events together with 
a happens-before partial order and a can-causally-affect relation, 
where some general axioms for system executions are satisfied. 
A proof of correctness is constructed by defining a mapping of sets of events to an abstract new event. 
This mapping induces a happens-before order and a can-causally-affect relation on the set of new abstract events. 
Hence, a new abstract (high-level) set of system executions are produced from a set of concrete (low-level) system executions. 
The proof is completed by showing that these high-level executions satisfy the specification. 
These high-level and low-level descriptions use the same mathematical language, 
allowing the abstraction process to be applied repeatedly.
Lamport's system executions framework cannot be easily adapted to weak memory consistency models. 
To more naturally capture weak memory consistency models, 
we typically use  more than one partial order in addition to agreement properties on these orders.

Lamport's system executions framework does not specify how these executions are generated.
Gischer \cite{DBLP:journals/tcs/Gischer88} and Pratt \cite{Pr91a} address this problem. 
The execution of an entire program is described as a collection of partially-ordered sets. 
These posets are generated by applying process algebra operations, 
such as sequential and parallel compositions, to smaller programs. 
Thus, a set of posets that represent all system executions can be recursively constructed. 
Then, simulation techniques are used to prove that the system executions match the specifications.

Many proofs of shared-memory algorithms assume Herlihy and Wing's strong consistency model, 
linearizability \cite{her90}.
A particularly useful and powerful property of linearizability is its \emph{locality}: 
proving separately that the implementation of each object in the system is linearizable 
implies  the correctness of the whole system implementation. 
Usually weak memory consistency models do not have such a locality property, complicating proofs of correctness.

Aspinal and Sevcik use a partial order formalism to represent the Java Memory Model (JMM) \cite{DBLP:conf/tphol/AspinallS07} in order to prove that data-race-free programs produce sequentially consistent executions on the Java Virtual Machine. 
The partial order constraints of JMM are used to produce a sequentially consistent total order.
Our partial order modeling shares similarity with theirs, but our low-level partial order constraints are used to produce computations that satisfy the partial order constraints of (the high-level) partition consistency, rather than a single total order.

%

\section{Definitions and Models}
\label{model.section}

An \emph{operation} consists of an \emph{operation invocation} and an \emph{operation response} often involving \emph{shared objects}. 
We use \emph{completed operation} to emphasize that an invocation is paired with a response.
A \emph{thread} generates a sequence of operation invocations in a sequence called \emph{program order}. 
A \emph{process} consists of a finite collection of threads.
A \emph{multiprogram} is a finite collection of processes. 

A \emph{computation} of the multiprogram is formed by \emph{arbitrarily} completing each operation invocation, 
in each individual thread sequence, with a response, creating a collection of sequences of completed operations.
Program order on operation invocations is naturally extended to 
define \emph{program order} on the set of completed operations of a computation. 
That is, a computation consists of a set of sequences of operations,
one sequence of operations for each thread of each process.
We denote this unrestricted set of computations of a multiprogram $P$ by \ComputationSet{P}. 
The subset of \ComputationSet{P} that could actually result from the execution of the multiprogram 
depends upon the distributed system's architecture. 
A \emph{memory consistency model} 
is a predicate defined on the set of all possible computations of a multiprogram; 
it filters these computations to include only those that could arise on the architecture being modeled. 
The subset of \ComputationSet{P} that satisfies the memory consistency predicate, \memCon{}, 
is denoted \FiltermemConVoid{P}.

We use the following notation, terminology and conventions for the rest of the paper. 
For a computation $C$ of a multiprogram $P$,
$O_C$  denotes all the operations of $C$.
A completed operation {\sc{oper}} with input $u$ that returns a value $v$ is denoted $\opReturn{}{oper}{u}{v}$.
For a set of operations $O$, $O|\cutwritesgroup{S}$ denotes the subset of all write operations to variables in $S$;
if $S$ is all the variables, we write $O|\cutwrites$;
$O|p $ denotes the subset of all operations by process $p \in P$. 
The \emph{program order relation on $O_C$}, denoted $\orderArrow{\programOrder_{C}}$, 
is the partial order formed by the union of the individual thread program orders\footnote{Since $p$ 
could be multithreaded, $(O_C|p, \orderArrow{\programOrder})$ is not necessarily a total order.}.
For all these notations we omit the subscript $C$ when it is obvious.
When we need the individual program order for a particular process or thread $p$, we write $\orderArrow{\programOrderp}$.
The style \Pred{pred}[args] is used to denote a predicate.
Given relations \orderArrow{R}, \orderArrow{T}, and a set $A$, define extension and agreement predicates on relations by:
\begin{itemize}
\item
$
\Pred{Extends}[ A , \orderArrow{R}  ,  \orderArrow{T}  ] \defequal  \forall a_1,a_2 \in A : a_1 \orderArrow{T} a_2 \implies a_1 \orderArrow{R} a_2 
$
\item
$
 \Pred{Agree}[ A , \orderArrow{R} , \orderArrow{T} ] \defequal \forall a_1,a_2 \in A : (a_1 \orderArrow{R} a_2) \Longleftrightarrow ( a_1 \orderArrow{T} a_2)
$.
\end{itemize}
Given a total order on a finite set, there is only one sequence of all the elements of the set that realizes that total order.
Therefore, we sometimes overload the term \emph{total order} for a finite set $A$: 
it refers to either the set of ordered pairs $(A, \orderArrow{T})$ in the order, 
or the sequence, which we denote by $T$, that realizes that total order. 
The notation  $\MapSet{ x_a : a\in A  }$
specifies a collection of items $x_a$, exactly one for each $a\in A$.

The most common shared objects for this paper are \emph{variables} with the sequential specification \cite{her90}:  
a sequence of \opreadVoid\ and \opwriteVoid\ operations on a variable $x$ is \emph{valid} if 
each $\opreadNoReturn{x}$ returns the value written by the most recent preceding 
$\opwrite{x,\cdot}$ in the sequence (or the initial value if no such \opwriteVoid\ exists). 
Other shared objects will be defined later as needed.
Any sequence of operations on a collection of objects is \emph{valid} if, for each object, 
the subsequence of operations on that object is valid.

The technical results of this paper concern three memory consistency models, called the partition consistency model,
the partial-order broadcast model, and the network model.   
We use these terms to describe the abstract machine that delivers the consistency guarantees,
but when we need to emphasize that these models are actually predicates on computations,
or when we need to denote them within other notation, 
we use the abbreviated predicate forms, \PredPCVoid, \PredPOBVoid, and \PredNWVoid\ respectively.

\paragraph{The partition consistency model} defines a class of abstract memory consistency models 
that is designed to capture processes that communicate by reading and writing shared variables.
It requires each process to ``see'' its own operations in addition to all other processes' writes in a valid total order. This order must extend program order. In addition, the views of all processes may be required to agree on the ordering of some specified subsets of operations.
%
More formally, 
let $K = \{ V_1, \ldots , V_m \}$ be a partition of a subset of the set $V$ of shared variables. 

\begin{definition}
\label{general-mcm-defn}
$
    \PredPC{K}{C}  \defequal  $ 
    $       \exists  \; \MapSet{  \text{valid total order} \; (O_C|p \cup  O_C|\cutwrites, \orderArrow{L_p}) : p\in \Processes  }  
$ satisfying \\
$(\forall p\in \Processes : \Pred{Extends}[ O_C|p \cup O_C|\cutwrites  , \orderArrow{L_p} ,  \orderArrow{\programOrder_{C}}  ] )  $
and $  ( \forall p,q \in P, i\in [1,m] :  \Pred{Agree}[ O_C|\cutwritesgroup{V_i} , \orderArrow{L_p} ,  \orderArrow{L_q} ] )$.
\end{definition}

Different instantiations of $K$ yield different memory consistency models including several well-known models. 
For example,  
Sequential Consistency requires that all processes agree on a single valid total order that extends program order. 
Thus, \Pred{SC} is \PredPC{\{V\}}{}. 
In the pipelined ramdom-access model every process ``sees'' all the writes of each other process in program order, 
but different processes can interleave these sequences differently, 
so there is no additional agreement beyond program order on the write operations. 
Thus, \Pred{P-RAM} is \PredPC{\emptyset}{}.
Goodman's Processor Consistency  requires that, in addition to \Pred{P-RAM}, for each shared variable,
processes agree on the order of all operations to that variable.
Thus \Pred{PC-G} is \PredPC{\Set{ \Set{v} ~|~  v \in V }}{}. 

A  variable is a \emph{single-writer variable} if it can be written by only one process,
otherwise it is a \emph{multi-writer variable}.
The multi-writer variable subset of $V$ is denoted $V|\cutmultiwriters$.
If $\Set{x} \in K$ and $x$ is a single-writer variable, then the \Pred{Agree} property for the set $\Set{x}$ 
holds automatically because write operations on $x$ are totally ordered by program order,
and program order is preserved by the \Pred{Extends} property.
Thus $\Set{x}$ can be removed from $K$ while maintaining \PredPC{K}{}.
Because implementations spend resources to maintain the consistency of each set in $K$, 
removing $\Set{x}$ from $K$ could reduce partition maintenance overhead in an implementation. 
This motivates two new natural instantiations of partition consistency, 
\begin{itemize}
\item
$\Pred{WeakPC-G} \defequal \PredPC{G}{}$ where the partition $G$ is given by 
$G = \{ \{v\} \ | \ v \in V|\cutmultiwriters \} $,  and 
\item
$\Pred{WeakSC} \defequal \PredPC{ \{V|\cutmultiwriters \} }{}$.
\end{itemize}
By the previous observation,
$\Pred{WeakPC-G}$ is equivalent to \Pred{PC-G};
however 
$\Pred{WeakSC}$ is strictly weaker than \Pred{SC} though still stronger than \Pred{PC-G}. 
Our preliminary investigation suggests that, for many programs,  $\Pred{WeakSC}$  is equivalent to \Pred{SC}. 
Yet, in our implementation, it can be substantially more efficient than \Pred{SC}.

\paragraph{The message-passing network model} captures a concrete reliable, message-passing asynchronous 
network of multi-threaded processes. 
Each process has a set of locally shared variables, which threads within that process use to 
 communicate with each other.
The accesses to locally shared variables are sequentially consistent\footnote{Weakening 
this assumption of ``local sequential consistency'' is possible. 
It only requires some  additional thread synchronization.   
Since this would add complication local to each process without otherwise changing the results of this paper, 
we do not include this option in the rest of this paper.}.
Threads of distinct processes communicate by sending and receiving messages, 
where messages from a sender to a receiver are received in the order sent.

This intuition of a network is formalized as follows. 
The shared objects are variables (shared between threads of the same process)
and messages (shared between different processes).  
Messages have distinct identifiers, and support the operations send 
\opsend{s,d,m} and the receive \oprecv{s,d,m}, 
where $s,d,m$ are the source, destination, and message contents respectively. 
A sequence of message operations is \emph{valid} if it contains at most one 
\opsendVoid{} and at most one \oprecvVoid{} of any message.  
We assume that for each local variable of a process, each write to that variable is distinct. 
(If this is not the case, the process can add sequence numbers to make them so.)
Define the following relations on the set $O_C$ of operations of a computation $C$:
\begin{description}
\item[Message causality:] 
(a message is received after it is sent) \\
$    x\orderArrow{\messageOrder_{C}} y\defequal  \ x,y\in  O_C \wedge x= \opsend{s,d,m} 
      \wedge y= \oprecv{s,d,m} $ 
\item[FIFO channel causality:] 
(two messages sent in program order to the same receiver are received in that order) \\
$    x\orderArrow{\fifoChannel_{C}} y \defequal \ x,y\in  O_C \wedge x=\oprecv{s,d,m} \wedge y= \oprecv{s,d,m'} 
     \wedge \opsend{s,d,m}\orderArrow{\programOrder} \opsend{s,d,m'}$
\item[Writes-into causality for variables:] 
(a value read from a variable must have been previously written to it)\\
$x \orderArrow{\writesInto_{C}}y\defequal \ x,y\in O_C \wedge x= \opwrite{w,z}\wedge y= \opread{w}{z}$
\item[Happens-before:] 
(operations happen in an order that observes the message, FIFO channel, and writes-into causalities) \\
$
\orderArrow{\happensBefore_{C}}\defequal 
(\orderArrow{\programOrder_{C}}\cup\orderArrow{\messageOrder_{C}}\cup\orderArrow{\fifoChannel_{C}}\cup\orderArrow{\writesInto_{C}})^{+}. \\
$
\end{description}
%
The definition \orderArrow{\happensBefore_{C}} is inspired by Lamport's happens-before \cite{lam78},
but that definition considers sequential processes that communicate only by message passing.
This definition adds shared memory communication between threads and is designed to incorporate weak consistency. 

\begin{definition}
$ \PredNW{C}\defequal
  \exists \; \MapSet{ \text{ valid total order } \; (O_C|p, \orderArrow{L_{p}})  : p\in \Processes }$ satisfying \\
   ($ \forall p\in P : \Pred{Extends}[O_C|p , \orderArrow{L_{p}} , \orderArrow{\happensBefore_C} ] $)
and $( \oprecv{s,d,m} \in O_C$ if and only if $\opsend{s,d,m} \in O_C ) $.
\end{definition}

This definition captures what we would expect of a reliable message-passing network that has 
FIFO channels between each pair of processors.
It requires that each process's view of its own operations is consistent with \orderArrow{\happensBefore_{C}} order. 
Thus, if two operations by threads of process $p$ are causally ordered,  
and even if the intermediate operations that cause that ordering are not visible to $p$, 
there must be a valid view of all $p$'s operations that does not conflict with that causal ordering. 
(Without this property, the model could, for example, allow a computation where 
process $p$ receives $m1$ from $q$, then does some local computation, then sends $m2$ to $q$,
and process $q$ receives $m2$ from $p$, then does some local computation, then sends $m1$ to $p$.
Since such a computation could not occur on a network, it should fail to satisfy the \PredNW{} memory consistency predicate.)
The last conjunct ensures that the received messages are exactly those that are sent. 

Any instance of the \PCName{} could be constructed directly on the message-passing network model. 
We obtain cleaner proofs and better abstraction, however, by introducing an intermediate level 
that isolates the fact that processes broadcast write updates and apply them locally
without the details of how broadcasting is managed.

\paragraph{The partial-order broadcast model} is designed to capture a collection of multithreaded processes, 
where threads within each process communicate through shared variables,
and updates are communicated between distinct processes 
using a one-to-all \opbcastVoid{} and a corresponding \opdeliverVoid{}.
Every process delivers updates in an order that extends the program order of the corresponding broadcasts.
Furthermore, updates can be labeled. 
Processes agree on the delivery order of all updates with the same label; 
such agreement is not required for differently labeled updates. 

The formal definition of this model uses 
variables (shared between threads of the same process) and
update objects (shared between different processes).  
Each update object is unique and supports the operation
\opbcast{u,l} (broadcast update $u$ with label $l$ to all) 
and the operation \opdeliver{u,l} (deliver the update $u$).
For unlabeled updates, the label has the null value, denoted $\bot$.
The delivery order of unlabeled updates is not constrained beyond program order.
A sequence of \opbcastVoid\ and \opdeliverVoid\ operations is \emph{valid} if 
1) no \opdeliverVoid\ precedes its corresponding \opbcastVoid\ 
(This restriction does not require the corresponding \opbcastVoid{} to be in the valid order), and,  
2) no specific \opdeliverVoid{} occurs more than once.
Define the deliver relation (a partial order) on the set $O_C$ of operations of a computation $C$ of the partial-order broadcast 
model by:\\
$x \orderArrow{\delorder_C} y \defequal x,y\in O_C ~ \wedge ~ x=\opdeliver{u_1 , l_1} ~ \wedge ~ y=\opdeliver{u_2 , l_2} ~ \wedge ~ \opbcast{u_1 , l_1} 
\orderArrow{\programOrder_{C}} \opbcast{u_2 , l_2} $.\\
Let $O_{C}|\cutdeliverslabel{l}$ denote the set of all \opdeliverVoid\ operations returning an update with label $l \neq \bot$.

\begin{definition}
\label{general-pob-defn}
$      \PredPOB{L}{C}  \defequal $
$         \exists \; \MapSet{ \text{valid total order} \; (O_C|p,\orderArrow{ { L_{{p}} } })   : {p}\in {P} }$  satisfying \\
$ (\forall p \in \Processes : \Pred{Extends}[ O_C|p  ,  \orderArrow{L_{p}} ,  \orderArrow{\programOrder_{p}} \cup \orderArrow{\delorder_{C}} ]) $
and $ (\forall p,q\in\Processes ,l \in L : \Pred{Agree}[ O_C| \cutdeliverslabel{l} , \orderArrow{L_p} ,  \orderArrow{L_q} ])$
and
$\forall p ( \opbcast{m, l} \in O_C$ if and only if $ \opdeliver{m, l} \in O_C|p )$.
\end{definition}

This definition captures what we described as the intermediate partial-order broadcast model. 
It requires that each process's view of its own operations is a valid sequence that extends 
the program order of its own threads and
delivers message according to the program order of the corresponding broadcasts.  
It also requires agreement between process's views of delivers of updates with the same label. 
The last conjunct ensures that every process delivers exactly the updates that were broadcast.


\section{Setup for Transformations and Proofs} 
\label{setup.section}

\subsection{Transformations setup}

We implement any partition consistency model $S$ on a \NWName\ model $N$  indirectly. 
We first implement $S$ on a \POBName\ model $T$ and then implement $T$ on $N$.
For each of these two steps we provide two different implementations, 
and prove each one is correct.
All four of the resulting proofs have similar structure and notation,
which is described in this section.

All our implementations transform code for a \emph{specified} model 
to code for a \emph{target} model. 
For clarity, {\sc small caps} font is used to denote specification level operations; 
{\tt Teletype} is used to denote target level operations.
To emphasize that a component belongs to the target level its name is sometimes
annotated with a ``hat'' as in $\widehat{\text{name}}$. 

\subsection{Proofs setup and structure}

Our transformations convert some specified multiprogram into a target multiprogram. 
To prove correctness, 
we must show the possible computations of these two multiprograms 
that can arise from their respective memory consistency models, have the same ``outcome''.  
We make this precise as follows.
Let \transform{\Processes} denote a transformation of multiprogram \Processes.
The possible computations of multiprogram \Processes\ (respectively, \transform{\Processes})
on the specified (respectively, target) memory consistent model \memCon\ (respectively, \MemCon) 
is the set \FiltermemConVoid{{\Processes}} 
(respectively, \FilterMemConVoid{{\transform{\Processes}}}).
But \transform{\Processes} transforms specified operation invocations that require a response 
into subroutines that return a response.  
So these returned responses can be used to \emph{interpret} each computation in
\FilterMemConVoid{{\transform{\Processes}}} as a computation of
\Processes.
We need to show that each such interpreted computation could have arisen in the specified model.
That is, we must show that 
the interpretation of any computation in \FilterMemConVoid{{\transform{\Processes}}} is in 
\FiltermemConVoid{\Processes}.
If this is satisfied for any \Processes, we say that \transform{ \cdot } \emph{correctly implements}
\memCon{} on \MemCon{}.
Figure \ref{correctnessPlan.figure}  depicts this proof obligation.

\begin{figure}[!h]

\begin{xy}
\xymatrix @C 1.0in {
{P} \ar[d]_{\tau} \ar@{-->}[r]_{\text{generates}} &   \FilterMemCon{\memCon{}}{P} &  \ar@{--}[l] _(0.2){\textnormal{show } \supseteq} I \\
  {\tau{}(P)}  \ar@{-->}[r]_{\text{generates}} &   \FilterMemCon{\MemCon{}}{\transform{P}} \ar @{-->}[ru]_{\textnormal{interpret}}&    
}
\end{xy}

\caption{Proof obligation for establishing correctness of transformations}
\label{correctnessPlan.figure}
\end{figure}
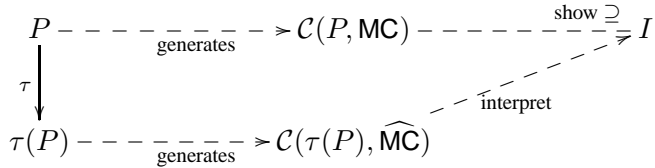

Each of our memory consistency models in Section \ref{model.section} is defined by requiring that for every computation
there is a collection of sequences of its operations such that 

\begin{itemize}
\item each sequence is valid,
\item each sequence extends some partial orders, 
and 
\item the set of sequences together satisfy some agreement constraints.
\end{itemize}

So we show that a computation $C$ satisfies a memory consistency model \memCon\
by constructing such a collection of sequences that jointly satisfy the \memCon's constraints.
Any such collection is said to \emph{witness} that $C$ satisfies \memCon, 
and is informally referred to as a collection of \emph{witness sequences}.
More formally, we use the 
predicate \WitOrd{A}{C}{\memCon} to assert that the collection of sequences $A$ witnesses that $C$ satisfies \memCon.

Let \Processes\ be a specified program and \transform{P} be a transformation of that program.
All our proofs have the following structure:\\
{\bf{}Assume:}
 $\widehat{C} \in  \FilterMemConVoid{{\transform{\Processes}}}$.
 Let $C \in \ComputationSet{P}$ be the interpretation of $\widehat{C}$. \\
{\bf{}Build:}
Choose any collection of sequences $\widehat{A}$ such that 
\WitOrd{\widehat{A}}{\widehat{C}}{\MemCon}.
Use $\widehat{A}$ to construct a corresponding collection of sequences $A$ for the operations in $C$. \\
{\bf{}Verify:}
Show that \WitOrd{A}{C}{\memCon}.

In this paper, we consider only \emph{finite} computations of the specified system that are completed in the target system.
For long-lived computations, we would need to consider computations that arise when a multiprogram is part way through its execution 
and extensions of such computations as the multiprogram continues to execute. 
Furthermore, the transformed multiprogram could be in a state where some processes are part way through executing the transformation of their
current operation invocation.  
Such an operation invocation is \emph{incomplete}.  
For example, a process could have sent some but not all messages required in the transformation of its current operation invocation, 
and messages sent by a process could be received by some recipients but not by others.
The problem of incomplete operations is taken care of by generalizing the technique used to show that a 
computation is linearizable as introduced by Herlihy and Wing \cite{her90}.
That is, in the {\bf{Assume}} step, 
we are allowed to adjust the computation $\widehat{C}$ so that every operation in the adjusted computation is 
complete before  proceeding with the {\bf {Build}} step that extracts the sequences $\widehat{A}$, 
and uses them to construct the witness sequences $A$.
This is done as follows. 
For every incomplete operation in $\widehat{C}$,  
either all its steps are erased or remaining steps are added so that it is complete. 
Such adding or erasing of steps could also change the operations of other processes 
since they may be receiving and acting on messages sent by operations that were incomplete.
So the steps of these operations are also either erased or completed.   
To take care of the problem that the multiprogram is long-lived, 
we must also show that the witness sequences constructed for a computation, say $C$, of the system 
are not messed up by the witness sequences that are constructed for an extension of that computation, say $C'$,
as the system continues to execute.
This is done by showing that the collection $A$ of sequences that are constructed to establish \WitOrd{A}{C}{\memCon}, 
are each prefixes of the corresponding collection $A'$ of sequences that are constructed to establish \WitOrd{A'}{C'}{\memCon}. 
For the proofs in this paper,  these two tasks are straightforward, but add considerable notational overhead.  
We leave it to the reader to observe that the results hold for long-lived computations but 
consider only finite computations with only completed transformations in this paper.

\subsubsection*{Proof diagrams}
\label{proof-diagrams}

When designing and debugging our proofs, 
we frequently used diagrams to record partial orders and various relationships between them.
Because these diagrams could be formalized and used to help make our proofs more precise 
and concise, we adopt this diagrammatic notation here. 
The notation use in these diagrams is as follows.
Let $a,b,c$ and $d$ be operations; 
and  $A$ and $B$ be set of operations. 
Edges in a diagram represent boolean expressions and a
diagram is interpreted as the conjunction of these expressions.
The basic building blocks are:

{\footnotesize
\noindent
\begin{tabular}{|l|l||l|l|}
\hline
diagram symbol &asserts & diagram symbol &asserts  \\ \hline
  \xymatrix{
a \ar@{=}[r] & b
} & $a=b$ 
&
 $
\xymatrix{
A \ar@{-(}[r] & B
}
$ & $A\subseteq B$\\
 $
\xymatrix{
  a \ar[r]^{L} & b
}
$ & $a\orderArrow{L}b$
&
 $
\xymatrix{
A \ar[r]^{L} & B
}
$ & 
  $\forall a\in A, b \in B : a\orderArrow{L}b$ \\
$\xymatrix{A \ar@{|-->}[r] & B}$ 
& $\exists b\in B : \forall a\in A: a\orderArrow{L}b$
& $\xymatrix{A \ar@{-->|}[r] & B}$ 
& $\exists a\in A : \forall b\in B: a\orderArrow{L}b$
\\
$\xymatrix{
a \ar@{--}[r]^R & b 
}$ & $a \simArrow{R} b$: where $R$ is relation 
&
\begin{xy}\xymatrix@R=5pt{
a \ar[rr]^{L}& \ar@{=>}[d] & b \\
c \ar[rr]_{M}& & d
}
\end{xy} & $(a \orderArrow{L} b) \implies (c\orderArrow{M} d)$ \\
\hline
\end{tabular}
}

\vspace*{-3pt}
Multiple edges are disjunctive; eg. $\xymatrix{a \ar@/_/[r]^{{\scriptscriptstyle D}} \ar@/^/[r]^{{\scriptscriptstyle E}} & b}$ 
asserts ($a\orderArrow{D} b$ or $a\orderArrow{E} b$).
%
%
The position of an arrow and its label are not significant; eg. 
$
\xymatrix{
a \ar[r]_L & b
}
$ and
 $
\xymatrix{
b & a \ar[l]_L
}
$ are equivalent.
Notice that for sets $A$ and $B$, the notation is similar to Lamport's system executions \cite{lam86b}.


\section{Implementing Partition Consistency on the Partial-Order Broadcast Model} 
\label{implement.section}

Partition consistency models processes that interact by reading and writing globally shared variables 
that have only weak consistency guarantees.
Recall that an instance of partition consistency has a partition $K$ of some of the variables and requires strong agreement within 
but not between sets of $K$.
Our task is to transform each specified process $p$ of such a system into a target process $\widehat{p}$ 
for the partial-order broadcast model, 
where inter-process communication is via the partial order broadcast primitive. 
Our implementation 
is a generalization of the way that totally ordered broadcast is used to implement sequential consistency \cite{attiya2004dcf}.
The predicate \PredPC{K}{} for partition consistency ensures 
agreement on the ordering of writes to all variables within the same class of the partition $K$;
the consistency predicate \PredPOB{L(K)}{} satisfied by the partial-order broadcast model is used to implement this agreement 
by enforcing agreement on the deliveries of updates with the same label. 

We achieve our implementation by:
\begin{itemize}
\item creating a label for each class in partition $K$; that is, $L(K) = \{i : V_i \in K\}$; 
\item mapping each $p$  to a thread $\targ{p}.m$ in the partial-order broadcast model; and 
\item augmenting each $\targ{p}.m$ with a companion delivery thread $\targ{p}.d$. 
\end{itemize}
There are two variants of the transformation.  
The pseudo-code for the slow-write/fast-read variant (SWFR) is shown in Figure \ref{PC-POB-impl-fig}.

\begin{figure}[h]
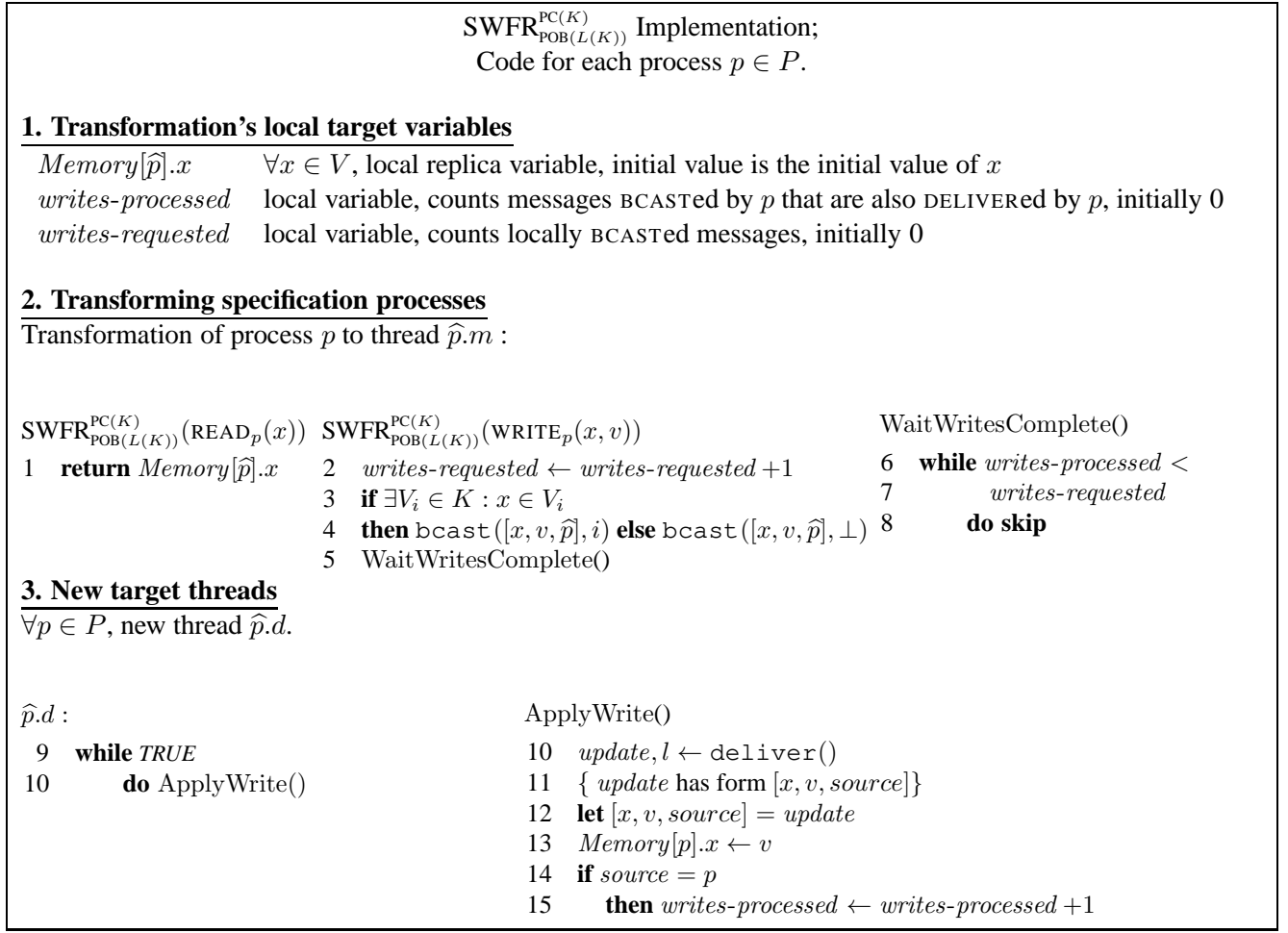


\medskip
\noindent
\fbox{
\begin{minipage}{6.75in}
\begin{center}
\swimpl{K}{L(K)}{} Implementation;\\
Code for each process $p \in P$.
\end{center}

\noindent
\textbf{\underline{1. Transformation's local target variables}}\\
\begin{tabular}{ll}
 $\id{Memory}[\widehat{p}].x$ &  $\forall x\in V$, local replica variable, initial value is the initial value of $x$ \\
$\id{writes-processed}$   &  local variable, counts messages \opbcastVoid{}ed by $p$ that are also \opdeliverVoid{}ed by $p$, initially 0 \\
$\id{writes-requested}$   &  local variable, counts locally \opbcastVoid{}ed messages, initially 0 \\ \\
\end{tabular}

\noindent
\textbf{\underline{2. Transforming specification processes}}\\ 
Transformation of process $p$ to  thread $\widehat{p}.m$ :

\noindent
{\small 
\begin{minipage}[t]{1.60in}
\begin{codebox}
\Procname{\swimpl{K}{L(K)}{\opreadpNoReturn{x}{p}}}
\li \Return $\id{Memory}[\widehat{p}].x$ \label{li:topmiddlereadend}
\end{codebox}
\end{minipage}
\begin{minipage}[t]{3in}
\begin{codebox}
\setlinenumberplus{li:topmiddlereadend}{1}
\Procname{\swimpl{K}{L(K)}{\opwritep{x,v}{p}}}
\li $\id{writes-requested} \gets \id{writes-requested}+1$
\li \If{$\exists V_i \in K : x \in V_i$}
\li \kw{then} $\opBcast{[x,v,\widehat{p}], i}$ \kw{else} $\opBcast{[x,v,\widehat{p}], \bot}$
\li \func{WaitWritesComplete}() \label{li:topmiddlewriteend}
\end{codebox}
\end{minipage}
\begin{minipage}[t]{2in}
\begin{codebox}
\setlinenumberplus{li:topmiddlewriteend}{1}
\Procname{\func{WaitWritesComplete}()}
\li \While{$\id{writes-processed} <$} 
\li \hspace*{25pt} $\id{writes-requested}$ 
\li \Do \kw{skip} \label{li:topmiddlewaitwritescompleteend}
\End
\end{codebox}
\end{minipage}
}

\noindent
\textbf{\underline{3. New target threads}}\\
$\forall p \in P$, new thread $\widehat{p}.d$.

\noindent
{\small 
\begin{minipage}[t]{0.4\textwidth}
\begin{codebox}
\setlinenumberplus{li:topmiddlewaitwritescompleteend}{1}
\Procname{$\widehat{p}.d :$}
\li     \While  {$\myconst{TRUE}$} 
\li \Do $\func{ApplyWrite}()$ \label{li:topmiddledeliverthreadend}
\End
\end{codebox}
\end{minipage}
\begin{minipage}[t]{0.4\textwidth}
\begin{codebox}
\setlinenumberplus{li:topmiddledeliverthreadend}{1}
\Procname{\func{ApplyWrite}()}
\li $\id{update}, l \gets \opDeliverNoReturn{}$
\li  $ \{ \textrm{ \id{update} has form } [x,v,source]  \}$
\li \kw{let} $[x,v,source] = \id{update}$
\li \Indent $\id{Memory}[p].x\gets v$
\li \If{$\id{source} = p $}
\li \Then $\id{writes-processed} \gets \id{writes-processed} + 1$
\End
\End
\end{codebox}
\end{minipage}
}

\end{minipage}
}
\medskip

\caption{Implementation of Partition Consistency on the Partial-Order Broadcast Model}
\label{PC-POB-impl-fig}
\end{figure}

The main thread, $\widehat{p}.m$, is derived from $p$ by replacing each \opreadVoid\ 
and each \opwriteVoid\ to a shared variable with a subroutine call. 
The transformation of a \opreadVoid\ simply returns the value stored in $\widehat{p}$'s local memory. 
The transformation of a \opwriteVoid\ creates a \opBcastVoid{} operation to be \opDeliverVoid{}ed to each target process.
It has a label corresponding to the partition class of the variable being written, if it exists.
The delivery thread, $\widehat{p}.d$, manages the \opDeliverVoid{} operations 
and maintains synchronization with $\widehat{p}.m$ via locally shared variables.
It repeatedly applies updates to the local memory it shares with $\widehat{p}.m$.
The procedure \func{WaitWritesComplete} causes $\widehat{p}.m$ to wait 
until $\widehat{p}.d$ has applied all the \opwriteVoid{}s previously broadcast by $\widehat{p}.m$. 

Under the $\swimpl{K}{L(K)}{}$ transformation,
each process has at most one outstanding local write,
since every write must be applied locally before the subroutine completes.
Every write contains a wait, making these writes ``slow''. 
An alternative is to move the \func{WaitWritesComplete} 
call from the end of the \opwriteVoid{} to the beginning of the \opreadVoid{}. 
This gives us a fast-write/slow-read (\fwimpl{K}{L(K)}{}) 
implementation.

\subsection{Correctness of the \swimpl{K}{L(K)}{} and \fwimpl{K}{L(K)}{} implementations}

The proofs of correctness of \swimpl{}{}{} and \fwimpl{}{}{} are very similar; they differ in only one step.
This step can be treated generically, so we present one proof for both implementations.
Let $\wimpl{K}{L(K)}{}$ refer to either implementation.

\begin{theorem}
\label{wrtrans-correctness-theorem}
Let \Processes\ be any multiprogram where each process in \Processes\ is a single-threaded program that accesses read/write variables
from a set $V$, and let $K =\{V_1, \ldots , V_k\}$ be any partition of a subset of $V$.
Then \wimpl{K}{L(K)}{\Processes} correctly implements \PredPC{K}{} on \PredPOB{L(K)}{}, for all such \Processes.
\end{theorem}

\begin{proof} \hspace{1cm}
\paragraph{Assume:}
Let $\widehat{C}$ be a computation in $\ComputationSet{\wimpl{K}{L(K)}{P}, \PredPOB{L(K)}{}}$ and let $C$ be the interpretation of $\widehat{C}$.
Let $\widehat{O}$ denote the set of operations $O_{\widehat{C}}$.
To show $\PredPC{K}{C}$ we construct witness sequences that satisfy the requirements of Definition \ref{general-mcm-defn}.

\paragraph{Build:}
Choose a collection of witness sequences $\MapSet{(\widehat{O}|\widehat{p}, \relLinearp{} ): \widehat{p} \in \wimpl{K}{L(K)}{P} }$.
That is, \\
\WitOrd{ \{\seqLinearp\ : \forall \widehat{p} \in  \widehat{P} \} }{\ComputationC}{\PredPOB{L(K)}{}}.

Construct a \emph{corresponding} set of sequences $\MapSet{(O_{C}|p\cup O_{C}|\cutwrites,\orderArrow{L_p}) : p\in P}$  as follows:
\begin{enumerate}
\item For each \opReadVoid\ or \opWriteVoid\ operation $o$ on a ``local replica'' variable we associate it with a specification level operation as follows: 
  \begin{enumerate}
  \item 
    The transformation sets up a one-to-one correspondence between the set of \opReadVoid\ operations of the \PredPOB{L(K)}{} system, and 
the \opreadVoid\ operations of the \PredPC{K}{} system.  
Specifically,  $\opRead{\id{Memory}[p].x}{v}$ in the implementation 
must have come from the transformation of a unique corresponding specification level $\opread{x}{v}\in O_C|p$.  
  \item
The transformation sets up, for each $\widehat{p}$, 
a one-to-one correspondence between the set of \opWriteVoid\ operations in $\widehat{O}|\widehat{p}$ of the \PredPOB{L(K)}{} system, and 
the \opwriteVoid\ operations of the \PredPC{K}{} system.
    More precisely, every $\opWrite{\id{Memory}[p].x,v}$ must have a $o_d = \opdeliver{[\const{write},x,v]}$ in the same $\func{ApplyWrite}()$ call.
    This deliver operation $o_d$ must have a corresponding $\opbcast{m}$,
    which can only have occurred in the transformation of some unique corresponding specification level write 
    $\wimpl{K}{L(K)}{\opwrite{x,v}}$. 

  \end{enumerate}
\item For each sequence \seqLinearp, 
build the sequence $Short(\seqLinearp)$ 
by removing all of operations that are applied to the broadcast object, 
and the local variable \id{writes-processed} and the local variable \id{writes-requested}.  
This leaves only the \opReadVoid\ and \opWriteVoid\ operations on the ``local replica'' variables in local memory.

\item  For each \seqLinearp, build a sequence $\seqlinearp$ by 
replacing each \opReadVoid\ (respectively, \opWriteVoid) operation in $Short(\seqLinearp)$ 
with the associated specification level \opreadVoid\ (respectively, \opwriteVoid) operation defined in step 1. 

\item
Notice that each sequence \seqlinearp\ contains exactly the operations in  $O_C|p\cup O_C|\cutwrites$. 
These sequences  induce the corresponding total orders  
$\MapSet{(O_C|p\cup O_C|\cutwrites, \orderArrow{L_p}) : p\in P }$.
\end{enumerate}

\paragraph{Verify:} 
We need to prove that \WitOrd{\{\seqlinearp\ : \forall p \in  \Processes\}}{\computationC}{\PredPC{K}{}}
holds for the orders \\
$\MapSet{(O_{C}|p\cup O_{C}|\cutwrites,\orderArrow{L_p}) : p\in P}$
constructed in the Build step.

First observe that removing from any valid sequence all operations related to specific objects preserves validity. 
Hence each $\func{Short}(\seqLinearp)$ is valid since \seqLinearp\ is valid.
Then replacing each \opReadVoid\ and \opWriteVoid\  with the corresponding \opreadVoid\ and \opwriteVoid\ respectively 
also clearly preserves validity.
We conclude that the  total orders $(O|p\cup O|\cutwrites, \orderArrow{L_p})$ are valid.
The following two lemmas establish the remainder of the proof:
 
\begin{tabular}{|l | r|}
\hline
Constraint & Lemma\\
\hline
$ \Pred{Extends}[O_C|p \cup O_C|\cutwrites, \orderArrow{L_p}, \orderArrow{\programOrder_C}]$ & Lemma \ref{top-middle-prog-extend} \\
$ \forall p,q \in P, i\in [1,k] :  \Pred{Agree}[O_C|\cutwritesgroup{V_i}, \orderArrow{L_p}, \orderArrow{L_q}] $ &  Lemma \ref{top-middle-group-ordering} \\
\hline
\end{tabular}\medskip
 
Therefore, $\PredPC{K}{C}$ as required.
\end{proof}

\begin{lemma}
\label{top-middle-prog-extend}
  $\forall p\in P :\Pred{Extends}[O|p \cup O|\cutwrites, \orderArrow{L_p}, \orderArrow{\programOrder_C}]$.
\end{lemma}

\begin{proof}
  Let $o_1 \orderArrow{\programOrder} o_2$ where $o_1, o_2 \in O|p \cup O|\cutwrites$. 
\PredPC{K}{} only has read/write variable objects, so there are four cases for $o_1,o_2$. 
Notice that if $o_1$ or $o_2$ is a \opreadVoid\ then $o_1, o_2 \in O|p$.

\subheader{Case 1: read, read}\\
Let $o_1=\opread{x}{v}$ and $o_2=\opread{y}{w}$.

\noindent
$\wimpl{K}{L(K)}{\opreadNoReturn{x}} \orderArrow{\ProgramOrder_p} \wimpl{K}{L(K)}{\opreadNoReturn{y}}
\implies \opRead{x}{v}$ \orderArrow{\ProgramOrder_p} $\opRead{y}{w}
\implies \\
\opRead{x}{v}$ \relLinearp{} $\opRead{y}{w} \implies o_1  \rellinearp{} o_2$.

\subheader{Case 2: read, write}\\
Let $o_1=\opread{x}{v}$ and $o_2=\opwrite{y,w}$. We have:

\noindent
 
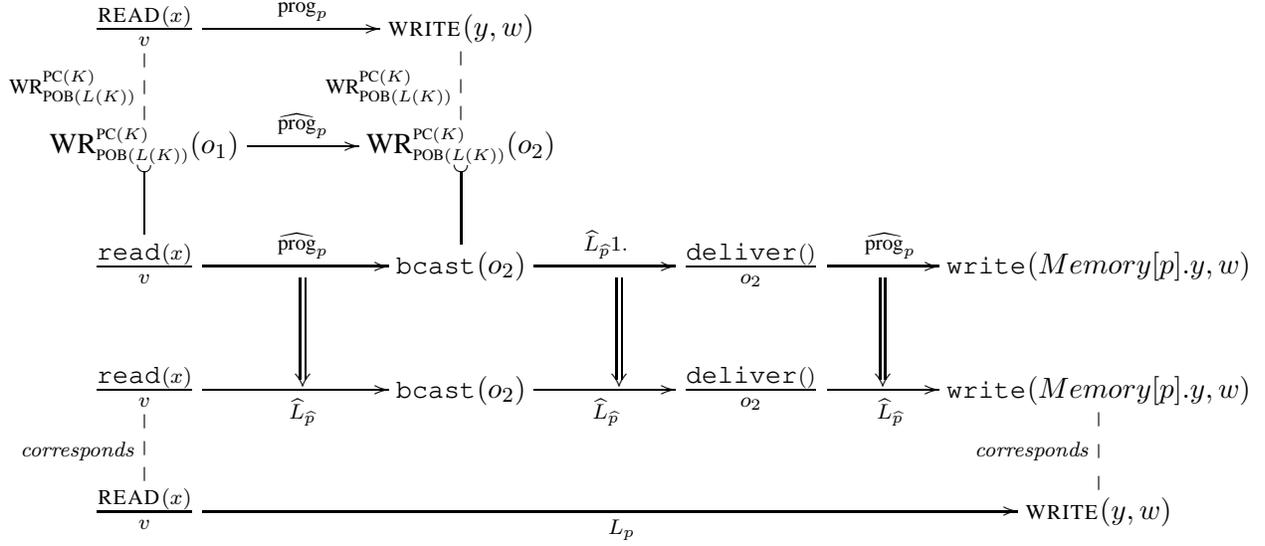
\begin{xy}
\xymatrix @C 0.25in { 
\opread{x}{v} \ar@{--}[d]_{\wimpl{K}{L(K)}{}} \ar[rr]^{\programOrder_p} & & \opwrite{y,w} \ar@{--}[d]_{\wimpl{K}{L(K)}{}} \\ 
\wimpl{K}{L(K)}{o_1} \ar@{)-}[d] \ar[rr]^{\ProgramOrder_p} & &  \wimpl{K}{L(K)}{o_2} \ar@{)-}[d]& & \\
\opRead{x}{v} \ar[rr]^{\ProgramOrder_p} & \ar@{=>}[d]  & \opBcast{o_2}  \ar[rr]_{}^{\seqLinearp 1.} & \ar@{=>}[d] & \opDeliver{o_2}  \ar[rr]^(0.4){\ProgramOrder_p} & \ar@{=>}[d] & \opWrite{Memory[p].y,w}  \\
\opRead{x}{v} \ar[rr]_{\seqLinearp} \ar@{--}[d]_{\id{corresponds}} & & \opBcast{o_2}  \ar[rr]_{\seqLinearp} & & \opDeliver{o_2}  \ar[rr]_(0.4){\seqLinearp} & & \opWrite{Memory[p].y,w} \ar@{--}[d]_{\id{corresponds}} \\
\opread{x}{v} \ar[rrrrrr]_{\seqlinearp} & & & & & & \opwrite{y,w}
}
\end{xy}
 \begin{enumerate}
   \item By definition of validity of update objects (since \opBcastVoid\ and \opDeliverVoid\ are both in $O|p$).
 \end{enumerate}

\subheader{Case 3: write, write}\\
Let $o_1=\opwrite{x,v}$ and $o_2=\opwrite{y,w}$. Suppose $o_1, o_2 \in O|q$. Then:

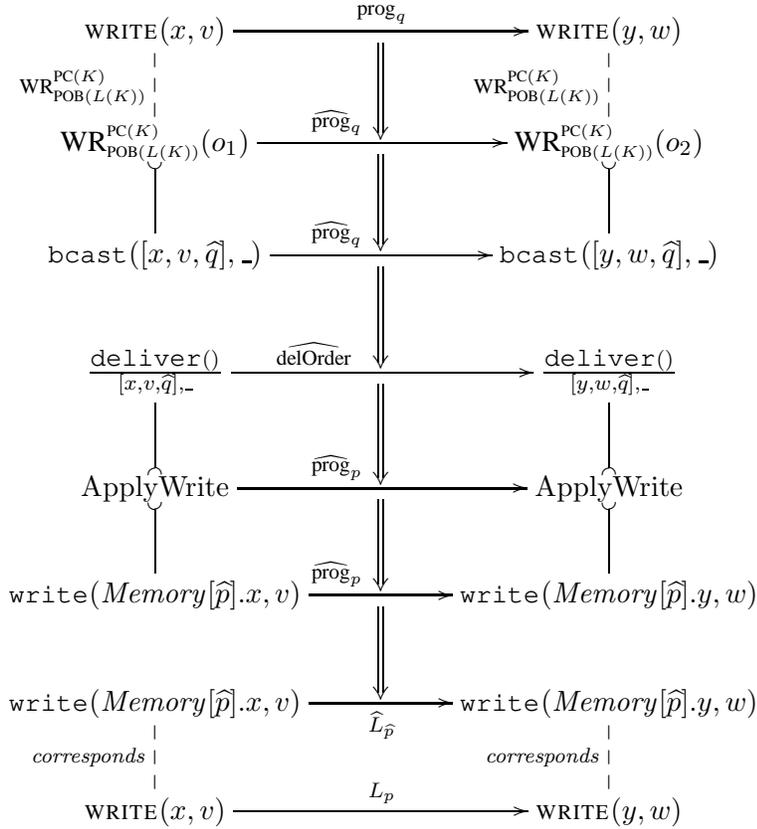
\[
\xymatrix{
\opwrite{x, v} \ar[rr]^{\programOrder_q} \ar@{--}[d]_{\wimpl{K}{L(K)}{}}& \ar@{=>}[d]  & \opwrite{y, w} \ar@{--}[d]_{\wimpl{K}{L(K)}{}} & \\
\wimpl{K}{L(K)}{o_1}  \ar@{)-}[d] \ar[rr]^(0.4){\ProgramOrder_q} &  \ar@{=>}[d] & \wimpl{K}{L(K)}{o_2} \ar@{)-}[d]& \\
\opBcast{[x, v, \widehat{q}],\_}  \ar[rr]^(0.4){\ProgramOrder_q} & \ar@{=>}[d]  & \opBcast{[y, w, \widehat{q}], \_} & \\
\opDeliver{[x, v, \widehat{q}], \_}    \ar[rr]^(0.35){\widehat{\delorder}} & \ar@{=>}[d] & \opDeliver{[y, w, \widehat{q}], \_}   &  \\
\func{ApplyWrite}  \ar@{)-}[u]  \ar[rr]^(0.4){\ProgramOrder_p} &  \ar@{=>}[d] & \func{ApplyWrite} \ar@{)-}[u]& \\
\opWrite{\id{Memory}[\widehat{p}].x, v} \ar@{-(}[u]  \ar[rr]^(0.4){\ProgramOrder_p} & \ar@{=>}[d]  & \opWrite{\id{Memory}[\widehat{p}].y, w}  \ar@{-(}[u] & \\
\opWrite{\id{Memory}[\widehat{p}].x, v} \ar@{--}[d]_{\id{corresponds}} \ar[rr]_{\seqLinearp} &   & \opWrite{\id{Memory}[\widehat{p}].y, w} \ar@{--}[d]_{\id{corresponds}}&\\
\opwrite{x, v}  \ar[rr]^{\seqlinearp} &   & \opwrite{y, w} &
}
\]

\subheader{Case 4: write, read } 
Let $o_1=\opwrite{x,v}$ and $o_2=\opread{y}{w}$. Then:

\begin{xy}
\xymatrix @C 0.225in {
 \opwrite{x, v} \ar[rrr]^{\programOrder_p} \ar@{--}[d]_{\wimpl{K}{L(K)}{}}&   & & \opread{y}{w} \ar@{--}[d]_{\wimpl{K}{L(K)}{}} & \\
\wimpl{K}{L(K)}{o_1}  \ar@{)-}[d] \ar[rrr]^{\ProgramOrder_{\widehat{p}}} &   & & \wimpl{K}{L(K)}{o_2} \ar@{)-}[dd]& \\
\opWrite{\id{writes-requested}, c_1} \ar[d]_{\ProgramOrder{}_{\widehat{p}}} \\ \opBcast{[x,v,\widehat{p}],\_} \ar[r]^{\ProgramOrder_{\widehat{p}}} \ar[dd]^{\seqLinearp}_{1.}   &  \opRead{\id{writes-processed}}{c_2} \ar[r]^{\ProgramOrder{}_{\widehat{p}}} & \opRead{\id{writes-requested}}{c_2 : c_2 \geq c_1 } \ar[r]^(0.6){\ProgramOrder_{\widehat{p}}} &   \opRead{y}{w}  \ar@{--}[ddd]_{\id{corresponds}} \\
 & \opWrite{\id{writes-processed}, c_1} \ar[u]^{3.}_{\seqLinearp} \\
  \opDeliver{[x,v,\widehat{p}], \_} \ar[r]_{\ProgramOrder_{\widehat{p}}}   &  \opWrite{x,v} \ar@{->}[u]^{2.}_{\ProgramOrder_{\widehat{p}}} \ar@{->}[uurr]^{\seqLinearp} \ar@{--}[d]_{\id{corresponds}} & \\
 &  \opwrite{x,v} \ar[rr]_{\seqlinearp} & & \opread{y}{w} 
}
\end{xy}

\begin{enumerate}
\item By validity of \seqLinearp.
\item Both fast-write and slow-write transformations call  \func{WaitWritesComplete} between any write and subsequent read by the same process. 
\item $c_2 \geq c_1$ and \id{writes-processed} is increased every time it is set. If a value greater than $c_1$ is read from \id{writes-processed}, it must have been written after $\opWrite{\id{writes-processed},c_1}$.
\end{enumerate}

Thus in all cases we have $o_1 \orderArrow{L_p} o_2$.
Therefore $\orderArrow{L_p}$ extends $\orderArrow{\programOrder}$.

\end{proof}

Case 4 exemplifies how the proofs for both \swimpl{}{}{} and \fwimpl{}{}{} are unified.
The only property needed is that there is a \func{WaitWritesComplete} between the \opbcastVoid{} and the \opread{}.
Both \swimpl{}{}{} and \fwimpl{}{}{} satisfy this property and so one proof suffices for the correctness of both implementations.

\begin{lemma}
\label{top-middle-group-ordering}
  $\forall p,q\in P, i\in [1,k] :  \Pred{Agree}[ O|\cutwritesgroup{V_i},  \orderArrow{L_p},  \orderArrow{L_q} ]$.
\end{lemma}

\begin{proof}
Let  $\opwriteVoid_1 , \opwriteVoid_2 \in O| \cutwritesgroup{V_i}$ for some $i$.
Then the transformation \wimpl{K}{L(K)}{\opwriteVoid_i} ($i \in \{1,2\}$) of each of these
contains a corresponding broadcast $\opBcast{msg_{i}, l}$ where the label $l$ is the same for each. 
For every process $p$,  \seqLinearp\ contains a \opDeliverVoid\ followed by a \opWriteVoid\ corresponding to each of these \opBcastVoid's. 
Let $w_{1p}, w_{1q}, w_{2p}, w_{2q} $ be these corresponding writes in $O_{\widehat{C}}$ for two processes $p$ and $q$ and 
suppose wolog that $w_{1p} \relLinearp{} w_{2p}$. 
Then 
$\opDeliver{ msg_{w_1} } \relLinearp{} w_{1p} \relLinearp{} \opDeliver{ msg_{w_2} } \relLinearp{} w_{2p} $ which implies 
$\opDeliver{ msg_{w_1} } \relLinearp[\widehat{q}]{} w_{1q} \relLinearp[\widehat{q}]{} \opDeliver{ msg_{w_2} } \relLinearp[\widehat{q}]{} w_{2q} $ 
because $\PredPOB{L(K)}{}$ requires that deliveries of messages with the same label must agree.
Hence, 
$\opwriteVoid_1 \orderArrow{L_p} \opwriteVoid_2$  and 
$\opwriteVoid_1 \orderArrow{L_q} \opwriteVoid_2$ by the construction of $L_p$ and $L_q$ from $\widehat{L}_{\widehat{p}}$ and $\widehat{L}_{\widehat{q}}$ 
respectively.
\end{proof}

\section{Implementing The Partial-Order Broadcast Model on the Message-Passing Network Model Using Tokens}
\label{token.section}

The partial-order broadcast model is very similar to the  message-passing network model. 
The \opreadVoid\ and \opwriteVoid\ operations are on local variables in both models, so they are
mapped by the transformation with the identity function.
That is, \opreadNoReturn{x} (respectively, \opwrite{x,v}) is mapped to 
 \opReadNoReturn{x} (respectively, \opWrite{x,v}). 
We need only specify how to implement $\opbcastVoid{}$ and \opdeliverVoid\ by \opSendVoid{}ing
and \opRecvVoid{}ing messages.
The processes in \Processes\ are numbered starting at 0, 
and organized into a virtual ring such that $ \func{next}(p) = (p + 1)\; {\bmod}\; \abs{\Processes} $. 
A token is created for each label $l\in L$, and for each token, a thread is created on each process to manage it.

\begin{figure}[!ht]
\medskip
\noindent
\fbox{
\begin{minipage}{6.75in}
\begin{center}
\tokenimpl{L}{} Implementation;\\
Code for each process $p \in P$.
\end{center}

\noindent
\textbf{\underline{1. Transformation's local target variables}}\\
\begin{tabular}{ll}
 $\id{needToken}[l]$ & for each $l\in L$; handshake boolean variable, initially  \myconst{TRUE} \\
 $\id{doorOpen}[l]$ & for each $l\in L$;  handshake boolean variable, initially  \myconst{TRUE} \\
$\widehat{x}$ &  $\forall x\in V$, replica of local variable, initial value is the initial value of $x$.

\end{tabular}

\noindent
\textbf{\underline{2. Transforming specification threads}}\\ 
Transformation of thread $p.m$ to $\widehat{p}.m$ and $p.d$ to $\widehat{p}.d$ :

\noindent
{\small 
\begin{minipage}[t]{3in}
  \begin{codebox}

    \Procname{$\tokenimpl{L}{\opbcast{ u , l }}$}
    \li \If{$l \ne \bot $}
    \li \Then  $\func{ProtectedBcast}(u , l )$  	
    \li \Else $\func{bcastop}(u,l)$ \label{li:tokenbcastend}
  \End
  \end{codebox}
\end{minipage}
\begin{minipage}[t]{2in}
  \begin{codebox}
\setlinenumberplus{li:tokenbcastend}{1}
    \Procname{$\tokenimpl{L}{\opdeliverNoReturn{}}$}
    \li $\codeRecvPattern{q,p,[\myconst{MESSAGE} , u , l ]}$
    \li \If $l \ne \bot$
    \li \Then $\opSend{p,q, [\myconst{ACK}]}$
\End
\li \Return $u , l$  \label{li:tokendeliverend}
  \end{codebox}
\end{minipage}
\begin{minipage}[t]{1.55in}  
\begin{codebox}
\setlinenumberplus{li:tokendeliverend}{1}
\Procname{$\tokenimpl{L}{\opreadNoReturn{x}}$}
\li \Return $\widehat{x}$ \label{li:tokenreadend}
\end{codebox}

\begin{codebox}
\setlinenumberplus{li:tokenreadend}{1}
\Procname{$\tokenimpl{L}{\opwrite{x,v}}$}
\li $\widehat{x} \gets v$ \label{li:tokenwriteend}
\end{codebox}
\end{minipage}
}

\noindent
{
\begin{minipage}[h]{0.4\textwidth}
  \begin{codebox}
\setlinenumberplus{li:tokenwriteend}{1}
    \Procname{\func{ProtectedBcast}($m$, $l$ )}

    \li $\id{needToken[l]}\gets \myconst{TRUE}$
    \li \While{$\neg \id{doorOpen[l]}$} \kw{skip}
    \li $\func{bcastop}( u , l )$
    \li $\id{doorOpen}[l] \gets \myconst{FALSE}$
    \li $\id{needToken[l]} \gets \myconst{FALSE}$ \label{li:tokenprotectedbcastend}
  \end{codebox}
\end{minipage}
\begin{minipage}[h]{0.4\textwidth}
  \begin{codebox}
\setlinenumberplus{li:tokenprotectedbcastend}{1}
    \Procname{\func{bcastop}($u , l $)} 
    \li \kw{forall} $q \in P $ 
    \li \Do $\opSend{p , q , [ \myconst{MESSAGE} , u , l]}$
\End
\li \If {$ l \ne \bot $ }
\li $\{\textrm{Wait for acknowledgment}\}$
\li \Then \kw{forall} $q \in P $ 
\li \Do $\codeRecvPattern{q , p , [\myconst{ACK}]}$ \label{li:tokenbcastopend}
  \End
  \End
  \end{codebox}
\end{minipage}
}

\noindent
\textbf{\underline{3. New target threads}}\\
One new token thread $\widehat{p}.\func{TokenThread}_l, \forall l \in L$:

\noindent
{\small 
 \begin{minipage}[t]{4.0in}
            \begin{codebox}
\setlinenumberplus{li:tokenbcastopend}{1}
           \Procname{$\targ{p}.\func{TokenThread}_{l}:$}
              \li \If{$\targ{p} = 0 $}
              \li \Then $\opSend{\targ{p},\func{next}(\targ{p}), [\myconst{TOKEN}, \myconst{BCASTGROUPTOKEN}_l]}$
            \End
            \li \kw{loop}
            \li \Do $\func{PassToken}_{\targ{p}}(l)$ \label{li:tokenthreadend}
          \End
            \end{codebox}
        \end{minipage}
\begin{minipage}[t]{1.5in}
  \begin{codebox}
\setlinenumberplus{li:tokenthreadend}{1}
    \Procname{$\func{next}(\widehat{p})$}
    \li \Return $(\widehat{p} + 1) \; \kw{mod} \; \left|\widehat{P}\right|$ \label{li:tokennextend}
  \end{codebox}
\end{minipage}

} 
\begin{minipage}[t]{0.5\textwidth}
  \begin{codebox}         
\setlinenumberplus{li:tokennextend}{1}
    \Procname{$\func{PassToken}_p(l)$}
    \li $\codeRecvPattern{q,p,[\myconst{TOKEN}, \myconst{BCASTGROUPTOKEN}_l]}$ 
    \li \If{\id{needToken[l]}} 
    \li \Then  $\id{doorOpen[l]} \gets \myconst{TRUE}$
    \li \While{\id{needToken[l]}} \kw{skip}

    \li $\opSend{p,\func{next}(p), [\myconst{TOKEN}, \myconst{BCASTGROUPTOKEN}_l]}$ \label{li:tokenpasstokenend}

\End

  \end{codebox}
\end{minipage}
\end{minipage}
}
\medskip

\caption{Token implementation of Partial-Order Broadcast on the Message-Passing Network Model}
\end{figure}

The broadcast of a labeled update requires that all processes agree on the delivery order of all
updates with the same label. 
To ensure this, 
the transformation of a \opbcastVoid\ of a labeled update has the process 
acquire the appropriate token for that label by synchronizing with its token thread,  
\opSendVoid\ a message containing the update information to each process, 
and wait for acknowledgments from all the processes before completing.
Since unlabeled updates only require that program order is maintained, 
unlabeled messages are sent (with \opSendVoid) to every other process without acquiring a token. 

To avoid deadlock this transformation requires that \opbcastVoid{}s and \opdeliverVoid{}s are invoked by separate threads.

Each call to \func{PassToken} manages one acquisition and subsequent release of a token.
It returns only after handshaking with \func{ProtectedBcast}
to determine that the token is no longer needed and can be released. 
$\codeRecvPattern{pattern}$ is pseudo-code that blocks until a message matching 
$pattern$ is received and stored in the appropriate pattern variables.

\subsection{Correctness of \tokenimpl{L}{} transformation}

\begin{theorem}
\label{tokentrans-correctness-theorem}
Let $\Processes$ be a multiprogram that uses \opreadVoid{}s, \opwriteVoid{}s, \opbcastVoid{}s, and \opdeliverVoid{}s 
where each process has two threads; 
one that calls \opdeliverVoid\ but not \opbcastVoid\
and one that  calls \opbcastVoid\  but not \opdeliverVoid. 
Then $\tokenimpl{L}{\Processes}$ correctly implements  \PredPOB{L}{} on \PredNW{},
for any label set $L$ and any such \Processes.
\end{theorem}

\begin{proof}\hspace{1cm}
\paragraph{Assume:} 
Let $\widehat{C}$ be a computation in $\ComputationSet{\tokenimpl{L}{\Processes}, \PredNW{}}$ 
and let $C$ be the interpretation of $\widehat{C}$.
Let $\widehat{O}$ denote the set of operations $O_{\widehat{C}}$.
To show $\PredPOB{L}{C}$ we construct witness sequences that satisfy the requirements of Definition \ref{general-pob-defn}.

\paragraph{Build:}
Choose some collection of witness sequences $\MapSet{(\widehat{O}|\widehat{p}, \orderArrow{\seqLinearp}): \widehat{p} \in  \tokenimpl{L}{P}}$.
That is, \\
\WitOrd{ \{\seqLinearp\ : \forall \widehat{p} \in  \tokenimpl{L}{P} \} }{\ComputationC}{\PredNWVoid}.
Recall that $\widehat{L_{\widehat{p}}}$ denotes the sequence induced by the total order $\orderArrow{\seqLinearp} $.
We now construct  $\MapSet{(O_{C}|p,\orderArrow{L_p}) : p\in P}$ from $\MapSet{(O|\widehat{p}, \orderArrow{\seqLinearp}): \widehat{p} \in  \tokenimpl{L}{P}}$ as follows. 
For each $\widehat{p}$, first create the sequence ${Short}(\seqLinearp)$ 
from $\widehat{L_{\widehat{p}}}$ by removing:
  \begin{itemize}
    \item all operations on the handshake variables $\id{needToken}[l]$ and $\id{doorOpen}[l]$ for all $l\in L$.
    \item all \opSendVoid{} operations except the first \opSendVoid\  of a  $\func{bcastop}()$ that is in $\widehat{L_{\widehat{p}}}$. 
    \item all \opRecvVoid{} operations except the \opRecvVoid\ of a \opdeliverVoid.
  \end{itemize}
Observe that each operation remaining in ${Short}(\seqLinearp)$ is a target level operation that 
was produced from the transformation of some specification level operation, 
rather than by a thread created by the transformation.
Now convert each ${Short}(\seqLinearp)$ to an new sequence $L_p$ of specification level operations by 
replacing each target level operations with the specification level operations that produced it. 
More precisely:

\noindent{}
  \begin{tabular}{|rl|c|}
    \hline
   Target Operation  in transformation & is replaced by \\
    \hline
    $\opRead{x}{v} \in   \tokenimpl{L}{\opread{x}{v}}$ & $\opread{x}{v}$.  \\
   $\opWrite{x,v}  \in  \tokenimpl{L}{\opwrite{x,v}}$ & $ \opwrite{x,v}$.  \\
   $\opSend{s,d,m} \in \tokenimpl{L}{\opbcast{m,l}}$  & $ \opbcast{m,l}$. \\ 
   $\opRecv{s,d,m}    \in \tokenimpl{L}{\opdeliver{m , l}}$ & $ \opdeliver{m , l}$.\\  
  \hline  
  \end{tabular}


\paragraph{Verify:}
The \opReadVoid\ and \opWriteVoid\ operations in $\seqLinearp$ are valid. 
Thus the subset that remains in ${Short}(\seqLinearp)$ remains valid because it consists of exactly the subset of these operations
that are applied to local variables, 
and projecting a valid sequence onto all the operations applied to a subset of objects preserves validity. 
Each \opReadVoid\ (respectively, \opWriteVoid) in ${Short}(\seqLinearp)$ is replaced with the corresponding 
\opreadVoid\ (respectively, \opwriteVoid) in the construction of \seqlinearp, so for each $p$,
all \opreadVoid\ and \opwriteVoid\ operations are valid. 
To see that the \opbcastVoid\ and \opdeliverVoid\ operations are also valid we must 
confirm that for each \seqlinearp, 
1) no \opdeliverVoid\ precedes its corresponding \opbcastVoid\ and,  
2) no specific \opdeliverVoid\ occurs more than once.
Sequence \seqLinearp (and thus ${Short}(\seqLinearp)$) is valid, 
so no \oprecvVoid\ precedes its corresponding \opsendVoid\, and no message is received more than once,
ensuring properties 1 and 2 after mapping from  ${Short}(\seqLinearp)$ to \seqlinearp.

The following table shows the properties we prove to verify the constraints of Definition \ref{general-pob-defn}:

\noindent{}
\begin{tabular}{| l | p{1.75in} |}
  \hline
   Constraint & Lemma \\
  \hline
 $\forall p\in P: (O_C|p, \orderArrow{\seqlinearp} ) $ is a valid total order & By construction of $L_p$ as proved above\\
 $\forall p\in P: \Pred{Extends}[O_C|p, \orderArrow{\seqlinearp{}}, \orderArrow{\programOrder_{p}}]$ & Lemma \ref{token-extend-program-lemma} \\
 $\forall p\in P: \Pred{Extends}[O_C|p, \orderArrow{\seqlinearp{}}, \orderArrow{\delorder_C}]$ & Lemma \ref{token-fifo-lemma}  \\
 $\forall p, q\in P, l\in L: \Pred{Agree}[O_C|\cutdeliverslabel{l}, \orderArrow{\seqlinearp{}}, \orderArrow{\seqlinearp[q]{}}]$ & Lemma \ref{token-agree-lemma}  \\
 $\forall p \in P: ( \opbcast{m, l} \in O_C$ if and only if $ \opdeliver{m, l} \in O_C|p )$ & by construction of \par $L_p$ and the network model. \\
  \hline
\end{tabular}\medskip

\end{proof}

\begin{lemma}
\label{token-extend-program-lemma}
 $\forall p\in P: \Pred{Extends}[O_C|p, \orderArrow{\seqlinearp{}}, \orderArrow{\programOrder_{p}}]$  \\
\end{lemma}

\begin{proof}
Let $o_1, o_2 \in O_C|p$ such that $o_1 \orderArrow{\programOrder_p} o_2$. 
Let $\tsimpl{L}{o}.\opCommon{low}{o}{}$ denote the operation in the transformation of $o$ 
that is mapped to $o$ when $Short{\seqLinearp}$ is converted to \seqlinearp.

\begin{xy}
\xymatrix{
o_1 \ar@{--}[d]_{\tokenimpl{L}{}} \ar[r]^{\programOrder{}} & o_2 \ar@{--}[d]^{\tokenimpl{L}{}}\\
\tokenimpl{L}{o_1} \ar[r]^{\ProgramOrder{}} & \tokenimpl{L}{o_2} \\
\tokenimpl{L}{o_1}.\opCommon{low}{o}{} \ar@{-(}[u] \ar[r]^{\ProgramOrder{}\cap\seqLinearp{}}_{1.} & \tokenimpl{L}{o_2}.\opCommon{low}{o}{} \ar@{-(}[u] \\
o_1 \ar@{--}[u]^{corresponds} \ar[r]^{\seqlinearp{}} & o_2 \ar@{--}[u]_{corresponds} \\
}
\end{xy}

\noindent 1. $\Pred{Extends}[O_C|\widehat{p}, \orderArrow{\seqLinearp}, \orderArrow{\ProgramOrder_p}]$
  
\end{proof}

\begin{lemma}
\label{token-fifo-lemma}  
 $\forall p\in P: \Pred{Extends}[O_C|p, \orderArrow{\seqlinearp{}}, \orderArrow{\delorder_C}]$  \\
\end{lemma}

\begin{proof}
Let $o_1, o_2\in O|p$  such that $o_1 \orderArrow{\delorder} o_2$.
Then $o_1 = \opdeliver{u_1,l_1}, o_2 = \opdeliver{u_2,l_2}$ and \\
$\opbcast{u_1,l_1} \orderArrow{\programOrder_s} \opbcast{u_2, l_2}$ for some process $s \in P$.

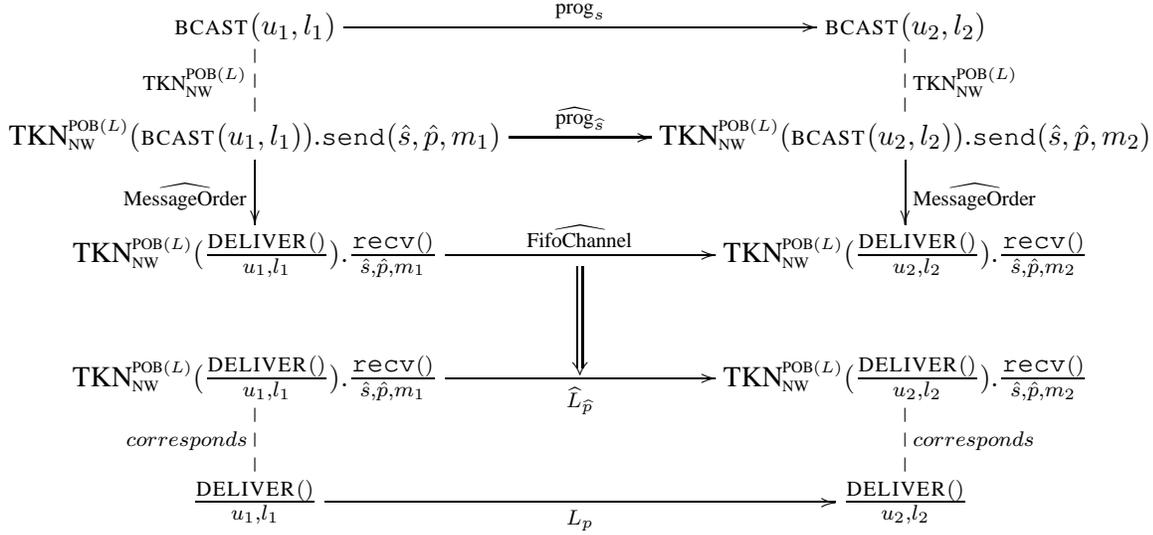
  \begin{xy}
    \xymatrix{
      \opbcast{u_1, l_1} \ar[rr]^{\programOrder_s} \ar@{--}[d]_{\text{\token}^{\POB{L}}_{\text{NW}}} & & \opbcast{u_2, l_2} \ar@{--}[d]^{\text{\token}^{\POB{L}}_{\text{NW}}} \\
      \tokenimpl{L}{\opbcast{u_1, l_1}}.\opSend{\hat{s}, \hat{p}, m_1} \ar[rr]^{\ProgramOrder_{\widehat{s}}} \ar[d]_{\MessageOrder} & & \tokenimpl{L}{\opbcast{u_2,l_2}}.\opSend{\hat{s}, \hat{p}, m_2} \ar[d]^{\MessageOrder} \\
      \tokenimpl{L}{\opdeliver{u_1, l_1}}.\opRecv{\hat{s}, \hat{p}, m_1} \ar[rr]^{\FifoChannel} & \ar@{=>}[d] & \tokenimpl{L}{\opdeliver{u_2, l_2}}.\opRecv{\hat{s}, \hat{p}, m_2} \\
      \tokenimpl{L}{\opdeliver{u_1, l_1}}.\opRecv{\hat{s}, \hat{p}, m_1} \ar[rr]_{\seqLinearp} & & \tokenimpl{L}{\opdeliver{u_2, l_2}}.\opRecv{\hat{s}, \hat{p}, m_2} \\
      \opdeliver{u_1, l_1} \ar[rr]_{\seqlinearp} \ar@{--}[u]^{corresponds} & & \opdeliver{u_2, l_2} \ar@{--}[u]_{corresponds}
    }
  \end{xy}

\end{proof}

The next two results are sublemmas that provide the pieces for our last requirement, Lemma \ref{token-agree-lemma}.
Informally, the first shows that the messages sent by a \opbcastVoid\ are all received (and acknowledged) before the \opbcastVoid\ completes.
The second ensures that an \opbcastVoid\ of an update with label $l$ by process $p$ is implemented while $\widehat{p}$ holds 
the token for label $l$.
These two facts together with the circulation of the token from process to process, 
combine to establish that all processes receive the update messages with label $l$ in the same order.

\begin{lemma}
\label{receive-within-bcast-lemma}
Each \opRecvVoid\ operation that corresponds to a specification level \opdeliverVoid\ operation is ordered by \HappensBefore\
order between the invocation and the response of the \opbcastVoid\ operation that matches this \opdeliverVoid.
\end{lemma}
\begin{proof}
Let \opRecv{\hat{q},\hat{p},[\myconst{MESSAGE} , u , l ]} correspond to some deliver \opdeliver{u, l}, and 
let \opbcast{u,l} be the matching \opbcastVoid. 
Consider the \opSend{\hat{p},\hat{q}, [\myconst{ACK}]} that is sent after this receive.

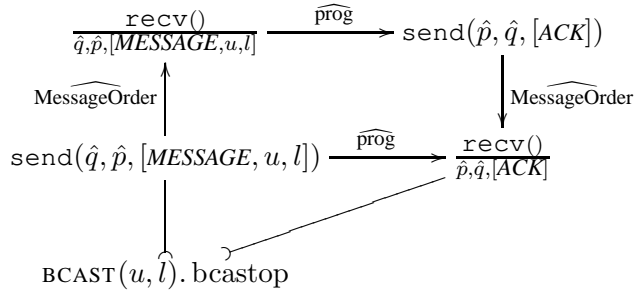
\begin{xy}
\xymatrix{
 \opRecv{\hat{q},\hat{p},[\myconst{MESSAGE} , u , l ]} \ar[r]^{\ProgramOrder} &  \opSend{\hat{p},\hat{q}, [\myconst{ACK}]} \ar[d]^{\MessageOrder{}} \\
\opSend{\hat{q},\hat{p},[\myconst{MESSAGE} , u , l ]} \ar@{-(}[d] \ar[u]^{\MessageOrder{}} \ar[r]^(0.625){\ProgramOrder} & \opRecv{\hat{p},\hat{q}, [\myconst{ACK}]} \ar@{-(}[ld]\\
\opbcast{u,l}.\func{bcastop}
}
\end{xy}

The Lemma follows because \HappensBefore\ order extends \ProgramOrder\ and \MessageOrder.   
\end{proof}

\begin{lemma}
\label{bcast-total-order-lemma}
For each $l\in L$, let $B_{l} = \{\tokenimpl{L}{\opbcast{u,l}}.\func{bcastop} : \opbcast{u,l} \in O_C \}$.
Then $(B_l , \orderArrow{\HappensBefore{}})$ is a total order.
\end{lemma}

\begin{proof}
Informally, for any label $l$, 
the main thread $\widehat{p}.\text{main}$ and the token thread for label $l$, $ \widehat{p}.\text{token}_l$,
handshake via the \id{needToken}$[l]$ and \id{doorOpen}$[l]$ variables.
This ensures that a \opbcastVoid\ of an update with label $l$ by process $p$ is implemented while $\widehat{p}$ holds that token.

More precisely: 

\begin{xy}
\xymatrix{
p.\text{main} & p.\text{token}_l \\
&  \opRecv{ [\myconst{TOKEN}, \myconst{BCASTGROUPTOKEN}_l]} \ar[d]^{\ProgramOrder{}} \\
\opWrite{\id{needToken}[l], \myconst{TRUE}} \ar[r]^{\WritesInto} \ar[d]_{\ProgramOrder{}} & \opRead{\id{needToken}[l]}{\myconst{TRUE}} \ar[d]^{\ProgramOrder{}} \\
\opRead{\id{doorOpen}[l]}{\myconst{TRUE}} \ar[d]_{\ProgramOrder{}} & \ar[l]_{\WritesInto} \opWrite{\id{doorOpen}[l], \myconst{TRUE}} \ar[ddd]^{\ProgramOrder{}} \\
\func{bcastop} \ar[d]_{\ProgramOrder{}} \\
\opWrite{\id{doorOpen}[l], \myconst{FALSE}} \ar[d]_{\ProgramOrder{}} \\
\opWrite{\id{needToken}[l], \myconst{FALSE}} \ar[r]_{\WritesInto} & \opRead{\id{needToken}[l]}{\myconst{FALSE}} \ar[d]^{\ProgramOrder{}} \\ 
 &  \opSend{ [\myconst{TOKEN}, \myconst{BCASTGROUPTOKEN}_l]} \\
}
\end{xy}

\HappensBefore\ order extends \WritesInto\ and \ProgramOrder\ and \MessageOrder.  
So the above proof diagram implies that
given a \func{bcastop}, \func{bcst}
$$
\ldots  
\opRecv{ [\myconst{TOKEN}, \myconst{BCASTGROUPTOKEN}_l]}
 \orderArrow{\HappensBefore{}}
\text{invocation}(\func{bcst}) 
 \orderArrow{\HappensBefore{}}
\text{response}(\func{bcst}) 
$$ 
$$
\orderArrow{\HappensBefore{}}
\opSend{ [\myconst{TOKEN}, \myconst{BCASTGROUPTOKEN}_l]} 
\ldots
$$ 
Any other \func{bcastop} with the same label $l$ by any process, 
must similarly be preceded by \\
\opRecv{ [\myconst{TOKEN}, \myconst{BCASTGROUPTOKEN}_l]}
and followed by 
\opSend{ [\myconst{TOKEN}, \myconst{BCASTGROUPTOKEN}_l]}.
Since there is exactly one token for label $l$, it follows by message validity 
that any two method calls to $\func{bcastop}$ for label $l$ are related by \HappensBefore\ order. 
\end{proof}

\begin{lemma}
\label{token-agree-lemma}
 $\forall p, q\in P, l\in L: \Pred{Agree}[O_C|\cutdeliverslabel{l}, \orderArrow{\seqlinearp{}}, \orderArrow{\seqlinearp[q]{}}]$. 
\end{lemma}

\begin{proof}
Let
$\opdeliver{u_1, l}, \opdeliver{u_2,l} \in O_C$ and let   $ \opbcast{u_1,l}, \opbcast{u_2,l}$ be, respectively, the 
\opbcastVoid\ operations that match these \opdeliverVoid\ operations. 
Consider the method calls 
$\func{bcst-1} = \tokenimpl{L}{\opbcast{u_1,l}}.\func{bcastop}$
and 
$\func{bcst-2} =\tokenimpl{L}{\opbcast{u_2,l}}.\func{bcastop}$ of the implementation of these \opbcastVoid s.

By Lemma \ref{bcast-total-order-lemma}, 
assume, without loss of generality, that
\func{bcst-1} \orderArrow{\HappensBefore} \func{bcst-2}. 

Then for any process $q$:

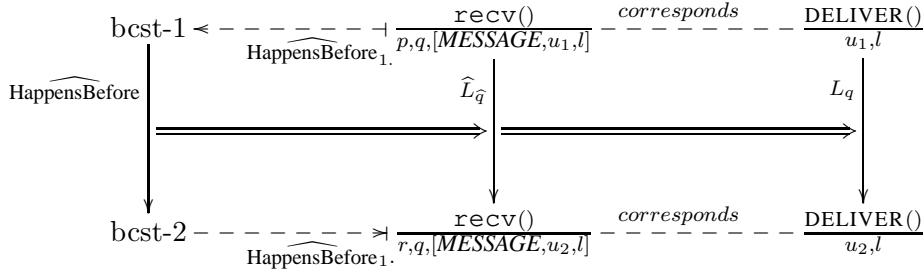
\begin{xy}
\xymatrix @C 1.0 in {
\func{bcst-1} \ar[dd]_(0.3){\HappensBefore}  & \ar@{|-->}[l]^{{\HappensBefore}_{1.}} \opRecv{p, q, [\myconst{MESSAGE} , u_1 , l ]}  \ar[dd]_(0.3){\seqLinearp[\widehat{q}]}  \ar@{--}[r]^{corresponds} & \opdeliver{u_1,l} \ar[dd]_(0.3){\seqlinearp[q]} \\
 \ar@{=>}[r] & \ar@{=>}[r] &  \\
\func{bcst-2} \ar@{-->|}[r]_{{\HappensBefore}{_1.}} & \opRecv{r, q, [\myconst{MESSAGE} , u_2 , l ]} \ar@{--}[r]^{corresponds} & \opdeliver{u_2,l}  &  \\
}
\end{xy}

1. By Lemma \ref{receive-within-bcast-lemma}.

Since this holds for every process $q$, all processes agree on the order of  \opdeliver{u_1,l}\ and \opdeliver{u_2,l}.
\end{proof}


\section{Implementing The Partial-Order Broadcast Model on the Message-Passing Network Model Using Timestamps}
\label{timestamp.section}

The \tsimpl{L}{} implementation (Figures \ref{TS-POB-NW-impl-fig} and  \ref{TS-POB-NW-impl-auxfig}) 
uses timestamps to enforce agreement on the order of deliveries of updates with the same label.
It generalizes Attiya and Welch's implementation of totally ordered broadcast \cite{attiya2004dcf}.

\opreadVoid\ and \opwriteVoid\ operations are mapped by the identity transformation to the \NWName{} model. 
The operations \opbcastVoid\ and \opdeliverVoid\ are implemented by
$\opSend{\mbox{\emph{source}}, \mbox{\emph{destination}},\mbox{\emph{message}} }$ 
and \opRecvNoReturn\ operations. 
No new threads need to be added in this implementation.

By the definition of \PredPOB{L}{}, 
the implementation must ensure that all processes agree on the deliver order of messages with the same label.
This is achieved using timestamps. 
Each process has $|L|$ priority queues, one for each message label. 
For unlabeled messages, it has one fifo-queue for each process.  
For priority-queues, we denote the enqueue and dequeue operations by
\opEnqueueVoid\ and \opDequeueVoid\ respectively.
For fifo-queues, we denote the enqueue and dequeue operations by
\opFIFOEnqueueVoid\ and \opFIFODequeueVoid\ respectively.

Labeled messages are handled as follows.
Messages with the same label are \opEnqueueVoid{}ed into the same priority-queue.
\opDequeueVoid{}  removes the message with the minimum (timestamp, process id) pair. 
To ensure that all processes \opdeliverVoid{} messages with the same label in the same order, 
the implementation guarantees that no message is dequeued by \opDequeueVoid{} before all messages 
with a smaller or equal timestamp have been received and \opEnqueueVoid{}ed. 
Processes keep their timestamps up to date by adopting the largest timestamp of any received message and 
\opSendVoid{}ing their updated timestamp to all other processes.

Unlabeled messages are handled slightly differently.
They are \opFIFOEnqueueVoid{}ed into the fifo-queue for the sending process, 
but timestamps are not used for \opFIFODequeueVoid{}ing because agreement of delivery order is not required for unlabeled messages.

The definition of \PredPOB{L}{} also requires that each process delivers messages in an order that extends the program order 
of the corresponding \opbcastVoid s. 
This is not automatically enforced because messages are spread across multiple priority-queues and fifo-queues. 
So the implementation uses counters.  
A  message is only delivered by a process if its counter value is 1 bigger than 
the counter value of the last delivered message from the same source.

\begin{figure}[!ht]
\medskip
\noindent
\fbox{
\begin{minipage}{6.75in}
\begin{center}
\tsimpl{L}{} :
Code for each $p\in P$
\end{center}

\noindent
\textbf{\underline{1. Transformation's local target variables}}

\begin{tabular}{ll}
\id{local-counter} & last broadcast counter value\\
 \id{counter[1..|P|]} & array of last received counter values, one for each process, initially all 0\\
\id{T[1..|P|]} & array of last received timestamp values, one for each process, initially all 0\\
\id{priorityQ[1..|L|]} &  array of priority-queues for messages, one for each label $l\in L$, initially all empty \\
\id{fifoQ[1..|P|]} &  array of fifo-queues for unlabelled messages, one for each process, initially all empty \\
$\widehat{x}$ & for each $x \in V$,  target-level name for $x$ \\
\end{tabular}

\noindent
\textbf{\underline{2. Transforming specification threads}}\\ 
Transformation of thread $p.m$ to $\widehat{p}.m$ and $p.d$ to $\widehat{p}.d$ :

\noindent
{\small
\begin{minipage}[t]{0.3\textwidth}
\begin{codebox}
\Procname{$\tsimpl{L}{\opreadNoReturn{x}}$}
\li \Return $\widehat{x}$ \label{li:timestampreadend}
\end{codebox}
\end{minipage}
\begin{minipage}[t]{0.3\textwidth}
\begin{codebox}
\setlinenumberplus{li:timestampreadend}{1}
\Procname{$\tsimpl{L}{\opwrite{x,v}}$}
\li $\widehat{x} \gets v$ \label{li:timestampwriteend}
\end{codebox}
\end{minipage}
\begin{minipage}[t]{0.6\textwidth}
  \begin{codebox}
\setlinenumberplus{li:timestampwriteend}{1}
    \Procname{$\tsimpl{L}{\opbcast{\id{update} , l}}$}
    \li  $\opSend{\widehat{p},\widehat{p}, [\blankfrac{\const{local-broadcast-request}, }{ \id{update} , l}]} $ \label{li:timestampbcastend}
  \end{codebox}
\end{minipage}
\begin{minipage}[t]{0.45\textwidth}
  \begin{codebox}
\setlinenumberplus{li:timestampbcastend}{1}
    \Procname{$\tsimpl{L}{\opdeliverNoReturn{}}$}
    \zi $\Set{\quad l \textrm{ can be } \bot \quad}$
    \li \While{$(\neg \exists l\in L : \func{CanExtract}(\id{priorityQ[l]}))$}
    \li $\wedge (\neg \exists \widehat{p} \in \widehat{P} : \func{CanDequeue}(\id{fifoQ}[\widehat{p}]))$
    \li \Do \func{HandleMessage()}
\End
\li \kw{case} (\kw{choose} $ (\const{L}l : \func{CanExtract}(\id{priorityQ[l]}))$
\li $| (\const{U} \widehat{p} : \func{CanDequeue}(\id{fifoQ}[\widehat{p}]))$)
\li \kw{of}
\li $(\const{L} l)$ \kw{then}  $qe \gets \opDequeueNoReturn{\id{priorityQ[l]}}$
\li $(\const{U} \widehat{p})$ \kw{then}  $qe \gets \opFIFODequeueNoReturn{\id{fifoQ}[\widehat{p}]}$
\li $counter[qe.src] \gets qe.counter$ \label{li:counterUpdate}
\li \Return $qe.\id{update} , l$ \label{li:timestampdeliverend}
  \end{codebox}
\end{minipage}
\begin{minipage}[t]{3.625in}
  \begin{codebox}
\setlinenumberplus{li:timestampdeliverend}{1}
    \Procname{$\func{HandleMessage}_p()$}

    \li $\widehat{s}, \widehat{p}, \id{message}\gets \opRecvNoReturn{}$

    \li \kw{case} \id{message} \kw{of}:
    \li $[\const{local-broadcast-request},\id{update}, l]$ \label{li:HandleLBRMess}
    \li \Then $T[\widehat{p}] \gets T[\widehat{p}]+1$ \label{li:incrementMyTS}
    \li  $\id{local-counter} \gets \id{local-counter} + 1$
    \li  $\id{queue-element} \gets [\id{update}, T[\widehat{p}], \id{local-counter} , \widehat{p}]$ \label{li:PrepQElement}
    \li  $\func{ProcessQueueElement}(\id{queue-element},l, \widehat{p})$ \label{li:ProcessQElementCallMe}
    \li  $\func{FifoBroadcast}([\const{ord-msg}, l, \id{queue-element}])$ \label{li:OrdBroadcast}
\End
\li  $[\const{ts-update}, \id{timestamp}, \widehat{q}]$ \label{li:HandleTSMess}
\li \Then  $T[\widehat{q}]\gets \id{timestamp}$ \label{li:incrementYourTS-tsmess}
\End
\li $[ \const{ord-msg} , l, \id{queue-element}]$ \label{li:HandleOrdMess}
\li \Then $T[\widehat{s}] \gets \id{queue-element}.\id{timestamp}$ \label{li:incrementYourTS-ordmess}
\li $\func{ProcessQueueElement}(\id{queue-element},l, \widehat{s})$ \label{li:ProcessQElementCallOther}
\li \If{$\id{queue-element}.\id{timestamp} > T[\widehat{p}]$} \label{queue-pob-update-condition}
\li \Then $T[\widehat{p}] \gets \id{queue-element}.\id{timestamp}  $ \label{li:boostMyTS}
\li $\func{FifoBroadcast}([\const{ts-update} , T[\widehat{p}] , \widehat{p}])$ \label{li:timestamphandlemessageend} \label{li:TSUpdateBroadcast}
\End

  \End
  \end{codebox}
\end{minipage}
}

\noindent
\begin{minipage}[t]{3.625in}

\end{minipage}
\begin{minipage}[t]{3in}

\end{minipage}

\noindent
\textbf{\underline{3. New target threads}}\\
No new target threads for this implementation.

\end{minipage}
}
\medskip

\caption{Timestamp Implementation of Partial-Order Broadcast on the Message-Passing Network Model}
\label{TS-POB-NW-impl-fig}
\end{figure}

\begin{figure}[!ht]
\fbox{
\noindent
\begin{minipage}{5.5in}
  
  \begin{codebox}
\setlinenumberplus{li:timestamphandlemessageend}{1}
    \Procname{$\func{CanExtract}(\id{priorityQ})$}
    \li \If{$\opIsemptyNoReturn{\id{priorityQ}}$}
    \li \Then \Return \const{false}
    \li \Else $qe \gets \opPeekminNoReturn{\id{priorityQ}}$
    \li \Return $(qe.counter = counter[qe.src] + 1)$  \label{li:CanExtractCounterTest}
    \zi  \quad $\wedge (\forall \widehat{q}\in \widehat{\Processes{}} : qe.ts \leq  T[\widehat{q}]) $ \label{li:timestampcanextractend}
  \End
  \end{codebox}
\begin{codebox}
\setlinenumberplus{li:timestampcanextractend}{1}
\Procname{$\func{CanDequeue}(\id{fifoQ})$}
\li \If{ \opFIFOIsEmptyNoReturn{\id{fifoQ}} }
\li \Then \Return \mfalse{}
\li \Else $qe \gets \opFIFOPeekHeadNoReturn{\id{fifoQ}} $ 
    \li \Return $(qe.counter = counter[qe.src] + 1)$ \label{li:timestampcandequeueend}
\End
\end{codebox}
   \begin{codebox}
\setlinenumberplus{li:timestampcandequeueend}{1}
    \Procname{$\func{ProcessQueueElement}(\id{queue-element}, l, \widehat{\id{source}})$}
    \li \If{$l \ne \bot$}
    \li \Then  $\opEnqueue{\id{priorityQ[l]}, \id{queue-element}}$ \label{li:priorityEnqueue} 
    \li \Else  $\opFIFOEnqueue{\id{fifoQ}[\widehat{\id{source}}], \id{queue-element}}$ \label{li:fifoEnqueue} \label{li:timestampprocessmessageend}
  \End
  \end{codebox}

\begin{codebox}
\setlinenumberplus{li:timestampprocessmessageend}{1}
\Procname{$\func{FifoBroadcast}(\id{message})$}
\li \For $\widehat{q} \in \widehat{P} \setminus \{\widehat{p}\}$
\li \Do \opSend{ \widehat{p} , \widehat{q} , \id{message} }
\End
\end{codebox} 
\end{minipage}
}
\caption{Auxiliary Functions for the Timestamp Implementation}
\label{TS-POB-NW-impl-auxfig}
\end{figure}

Observe that the \opdeliverVoid{} transformation does the heavy lifting in this implementation.
In order to avoid race conditions and more complicated synchronization, 
this implementation requires that for each process, at most one thread performs \opdeliverVoid{}.
Priority-queues and fifo-queues can be constructed from just variables 
because each queue is accessed by only one thread. 

In Figures \ref{TS-POB-NW-impl-fig} and \ref{TS-POB-NW-impl-auxfig}, messages are designated as one of three types:
$\const{local-broadcast-request}$, $\const{ts-update}$, $\const{ord-msg}$.
Depending on type, a message can contain a timestamp,  counter,  sender id (denoted $m.src$) 
as well as the label and value for the requested update. 

\subsection{Correctness of the \tsimpl{}{} implementation}

\begin{theorem}
\label{timestampAlg.theorem}
Let $\Processes$ be any multiprogram that uses \opreadVoid{}s, \opwriteVoid{}s, \opbcastVoid{}s, and
 \opdeliverVoid{}s such that at most one thread in each process calls \opdeliverVoid{}.
Then  \tsimpl{L}{\Processes} correctly implements $ \PredPOB{L}{}$ on $\PredNW{}$, for any label set $L$ and any such \Processes.
\end{theorem}

\begin{proof}\hspace{1cm}
\paragraph{Assume:} 
Let $\widehat{C}$ be a computation in $\ComputationSet{\tsimpl{L}{\Processes}, \PredNW{}}$ and let $C$ 
be the interpretation of $\widehat{C}$.
Let $\widehat{O}$ denote the set of operations $O_{\widehat{C}}$.
To show \PredPOB{L}{C}, we construct witness sequences that satisfy the requirements of Definition \ref{general-pob-defn}.

\paragraph{Build:}
Choose a collection of witness sequences $\MapSet{(\widehat{O}|\widehat{p}, \relLinearp{} ): \widehat{p} \in \tsimpl{L}{P} }$.
That is, \\
\WitOrd{ \{\seqLinearp\ : \forall \widehat{p} \in  \widehat{P} \} }{\ComputationC}{\PredNW{}}.
Recall that we use $\seqLinearp{}$ to denote the sequence induced by  $(\widehat{O}|\widehat{p}, \orderArrow{ \seqLinearp{}})$. 

Construct the sequence ${Short}(\seqLinearp)$ from $\seqLinearp{}$ by removing:
  \begin{enumerate}
  	\item all the operations on the $\id{T}$, $\id{counter}$, and $\id{local-counter}$ variables.
  	\item all \opSendVoid{} and \opRecvVoid{} operations except the \opSendVoid{} of a  $[\const{local-broadcast-request}]$.
  	\item all operations on priority-queues and fifo-queues except \opDequeueVoid\ and \opFIFODequeueVoid\ operations. 
  \end{enumerate}
The operations remaining in ${Short}(\seqLinearp)$ are \opReadVoid s, \opWriteVoid s, \opSendVoid s, \opDequeueVoid s
and \opFIFODequeueVoid s.
 Each such operation, $\opCommon{low}{op}$, was produced from a  transformation of some unique specification level operation, denoted \lift{$\opCommon{low}{op}$}, of the main thread.
Specifically: \\
\lift{$\opRead{x}{v}$} $=$ $\opread{x}{v}$ \\
\lift{$\opWrite{x,v}$} $=$ $\opwrite{x,v}$ \\ 
\lift{$\opSend{s,d,[\const{local-broadcast-request}, \id{update}, l]}$} $=$ $\opbcast{\id{update}, l}$ \\
\lift{$\opDequeue{\id{priorityQ}[l]}{m}$} $=$ $\opdeliver{m ,l}$ \\
\lift{$\opFIFODequeue{\id{fifoQ}[p]}{m}$} $=$ $\opdeliver{m ,\bot}$ 
(where $\bot$ denotes an unlabelled update). \\
For a sequence $S$, \Lift{$S$} denotes the sequence formed by applying \lift{} to each operation in $S$.
Define the sequence \seqlinearp{} to be \Lift{$Short(\seqLinearp)$}. 

\paragraph{Verify:}
We now complete the proof by showing that the sequences $\MapSet{ \seqlinearp{} : p \in \Processes}$ just constructed 
are witness sequences for $\PredPOB{L}{C}$. We do this by verifying the constraints of Definition \ref{general-pob-defn}. 

\noindent
\begin{tabular}{| l | l |}
  \hline
   Constraint of Definition \ref{general-pob-defn} & Lemma \\
  \hline
 $\forall p\in P: (O_C|p, \orderArrow{\seqlinearp} ) $ is a valid total order & Lemma \ref{valid-lemma} \\
 $\forall p\in P: \Pred{Extends}[O_C|p, \orderArrow{\seqlinearp{}}, \orderArrow{\programOrder{}}]$ & Lemma \ref{timestamp-extend-program-lemma} \\
 $\forall p\in P: \Pred{Extends}[O_C|p, \orderArrow{\seqlinearp{}}, \orderArrow{\delorder_C{}}]$ & Lemma \ref{timestamp-fifo-lemma} \\
 $\forall p, q\in P, l\in L: \Pred{Agree}[O_C|\cutdeliverslabel{l}, \orderArrow{\seqlinearp{}}, \orderArrow{\seqlinearp[q]{}}]$ & Lemma \ref{timestamp-agree-lemma} \\
$ \forall p\in P: (\opbcast{m, l} \in O_C$ if and only if $\opdeliver{m, l} \in O_C|p$ ) & Lemma \ref{bcast-deliver-lemma}  \\
  \hline
\end{tabular}

\end{proof}

Several parts of the following proofs are the same for labeled and unlabeled updates.
When this is the case, 
we use ``queue'' to mean any priority-queue or fifo-queue. 
We use \opEitherEnqueueVoid\ to denote either a \opEnqueueVoid\ applied to \id{priorityQ[l]} for some label $l$,
or a \opFIFOEnqueueVoid\ applied to \id{fifoQ}$[\widehat{p}]$ for some process $\widehat{p}$. 
Similarly, \opEitherDequeueVoid\ denotes either an \opDequeueVoid\ or a \opFIFODequeueVoid. 
A subscript on a local variable indicates which process owns that variable.  
For example, $T_{\widehat{p}}[\widehat{q}]$ denotes $\widehat{p}$'s variable  
$T[\widehat{q}]$. 
Similarly, a subscript on an operation indicates which process applied the operation. 
For example, 
\opEnqueuep{\id{priorityQ,m}}{\targ{p}} 
denotes that this \opEnqueueVoid\ was applied by \targ{p}.

We begin with three sublemmas that capture the essential properties of  timestamps.  
We rely on these lemmas later. 

\begin{lemma}
\label{T-values-increasing-lemma}
For all processes \targ{p} and \targ{r}, 
the writes to $T_{\widehat{p}}[\widehat{r}]$ taken in program order have strictly increasing values. 
\end{lemma}

\begin{proof}
$T_{\widehat{p}}[\widehat{p}] $ changes value only in Line \ref{li:incrementMyTS} where it is incremented, 
and Line  \ref{li:boostMyTS} where it is boosted to a bigger value.
Therefore $T_{\widehat{p}}[\widehat{p}] $ never decreases. 

For $\widehat{r} \neq \widehat{p}$, 
$\widehat{p}$ writes a new value $t$ to $T_{\widehat{p}}[\widehat{r}]$ only
in Lines \ref{li:incrementYourTS-tsmess} and \ref{li:incrementYourTS-ordmess}
because $\widehat{p}$ received a \const{ts-update} or \const{ord-msg} message from $\widehat{r}$ with timestamp $t$. 
So consider the timestamps in \const{ts-update} and \const{ord-msg} messages sent by $\widehat{r}$ to $\widehat{p}$. 
We have just seen that  $T_{\widehat{r}}[\widehat{r}]$ does not decrease, 
so, given the increment in Line \ref{li:incrementMyTS}, 
any \const{ord-msg} sent by $\widehat{r}$ to $\widehat{p}$ (Line \ref{li:OrdBroadcast})
contains a strictly bigger timestamp than that of any previous message sent by $\widehat{r}$ to $\widehat{p}$.
Similarly, given the increase in Line \ref{li:boostMyTS}, 
any \const{ts-update} message sent by \targ{r} (Line \ref{li:TSUpdateBroadcast})
contains a strictly bigger timestamp than that of any previous message sent by \targ{r} to \targ{p}.
Since messages are received in fifo order, 
these messages that $\widehat{p}$ receives from $\widehat{r}$ have increasing timestamps,
confirming that each of $\widehat{p}$'s writes to $T_{\widehat{p}}[\widehat{r}]$ writes a bigger value than was previously written.
\end{proof}

\begin{lemma}
\label{timestamp-deliver-key-lemma}
If $m_1$ and $m_2$ are  \const{ord-msg} messages with labels $g$ and $h$ ($h \neq \bot $)
and \id{queue-element}s $qe_1$ and $qe_2$ respectively 
such that $qe_1.ts \le qe_2.ts $ 
then for all $\widehat{p}\in\widehat{P}$, \\ 
$ \opEitherEnqueuep{\id{queue_g},qe_1}{\targ{p}} \orderArrow{\ProgramOrder{}} \opDequeuep{\id{priorityQ[h]}}{qe_2}{\targ{p}}$.
\end{lemma}

\begin{proof}

For a process $\widehat{p}$ in $\widehat{P}$ to execute $\opDequeuep{\id{priorityQ[h]}}{qe_2}{\targ{p}}$,
$\func{CanExtract}_{\targ{p}}({\id{priorityQ[h]}})$ must have returned \const{true}, 
implying $T_{\widehat{p}}[\widehat{q}] \geq qe_2.ts$ for every process $\widehat{q}$, 
and hence, for $qe_1.src$.
By Lemma \ref{T-values-increasing-lemma}, 
each write to $T_{\widehat{p}}[qe_1.src]$ is an increasing value, so
$T_{\widehat{p}}[{qe_1.src}] \geq qe_2.ts$ remains true. 

Notice that each \opEitherEnqueuep{queue_g, qe_1}{p} is called from either Line \ref{li:ProcessQElementCallMe} or Line 
\ref{li:ProcessQElementCallOther},
and each is preceded by a write of $qe_1.ts$ to $T_{\widehat{p}}[qe_1.src]$
(Lines \ref{li:incrementMyTS} and \ref{li:incrementYourTS-ordmess}).
Thus:

\noindent
\begin{xy}
\xymatrix {
\opWritep{T_{\targ{p}}[qe_1.src],qe_1.ts}{\targ{p}} \ar^{\ProgramOrder{} \; 1.}[rr]  \ar^{\ProgramOrder}[dr] & \ar@{=>}[d] & \opReadp{T_{\targ{p}}[qe_1.src]}{val}{\targ{p}} \ar^{\ProgramOrder}_{2.}[dd] \\
& \opEitherEnqueuep{queue_g, qe_1}{\targ{p}}  \ar^{\ProgramOrder}[ur] & \\
 1.\; val \geq qe_1.ts  &  2.\; val \geq qe_2.ts \text{ by code}    & \opDequeuep{\id{priorityQ[h]}}{qe_2}{\targ{p}} \\ 
}
\end{xy}

\end{proof}

\begin{lemma}
\label{timestamps-catchup-elements-lemma}

If any process \opEitherEnqueueVoid s a queue-element with timestamp $ts$, then 
for all processes $\widehat{q}$ and $\widehat{r}$ eventually the value for $T_{\targ{q}}[\targ{r}]$ 
becomes and remains at least $ts$. 
\end{lemma}
\begin{proof}

By Lemma \ref{T-values-increasing-lemma}, $T_{\targ{q}}[\targ{r}]$ never decreases, 
so it suffices to show that it eventually takes on a value equal to at least $ts$.

A queue-element, say $qe = [u, ts, c, \targ{p}]$ , can be  \opEitherEnqueueVoid ed by process $\targ{p}$ in line  
\ref{li:ProcessQElementCallMe} 
or by process $\targ{q}  \neq \targ{p} $ in line \ref{li:ProcessQElementCallOther} of \func{HandleMessage}. 
Even in the second case, however, $qe$ was previously \opEitherEnqueueVoid ed by process $\targ{p} $ in line  
\ref{li:ProcessQElementCallMe}.
Process $\targ{p}$  \opEitherEnqueueVoid s $qe$ 
as a consequence of its own \func{LocalBraoadcastRequest} 
and incremented $T_{\widehat{p}}[\widehat{p}] $ to equal $ts$ in Line \ref{li:incrementMyTS}.
It then broadcasts $qe$ to every other process. 
For each other process \targ{r}, when $\targ{r}$ receives $qe$, 
if $T_{\widehat{r}}[\widehat{r}] $ is smaller than $ts$, then it is boosted in Line \ref{li:boostMyTS} to equal $ts$.
Thus, for every $r$,  $T_{\widehat{r}}[\widehat{r}] $ is eventually at least as big as $ts$. 
Furthermore, $\widehat{r}$ \  \func{FifoBroadcast}s every change of $T_{\widehat{r}}[\widehat{r}] $ 
via  either an \const{ord-msg} at Line  \ref{li:OrdBroadcast}
or a \const{ts-update} message at Line \ref{li:TSUpdateBroadcast},
which
upon receipt, by each other process $\widehat{q}$,  
causes $\targ{q}$ to set $T_{\widehat{q}}[\widehat{r}] $ to the received value
(Lines \ref{li:incrementYourTS-tsmess} and \ref{li:incrementYourTS-ordmess}).
It follows that for all processes $\widehat{q}, \widehat{r} \in \widehat{P}$, eventually $\widehat{q}$'s 
value for $T[\widehat{r}] $ is at least $ts$.

\end{proof}

\begin{lemma}
\label{bcast-deliver-lemma}
$ \opbcast{\id{update}, l} \in O_C$ if and only if $\opdeliver{\id{update}, l} \in O_C|p, \forall p$. 
\end{lemma}

\begin{proof}
Each transformation of a \opbcast{\id{update},l} by process $p$ generates a unique \const{local-broadcast-request} by $\widehat{p}$, 
the transformation of $p$. 
Each \const{local-broadcast-request} results in the preparation of a \id{queue-element}, say $\eta$, (Line \ref{li:PrepQElement})
that contains \id{update}, and which is a parameter in the call by  $\widehat{p}$ 
to $\func{ProcessQueueElement}$ (Line \ref{li:ProcessQElementCallMe}).
This call results in an \opEitherEnqueueVoid\ to \id{priorityQ[l]} if $l \neq \bot$ 
(Line \ref{li:priorityEnqueue})
or to \id{fifoQ}$[\widehat{p}]$ if $l = \bot$
(Line \ref{li:fifoEnqueue}). 

Process $\widehat{p}$ next sends a copy of $\eta$ 
to every other process (Line \ref{li:OrdBroadcast}) in \func{HandleMessage}. 
Therefore, each \opbcast{\id{update},l} by process $p$ results in an \opEitherEnqueueVoid\ of $\eta$ at $\widehat{p}$,
and also results in an \opEitherEnqueueVoid\ of $\eta$ at every other remote process:\\

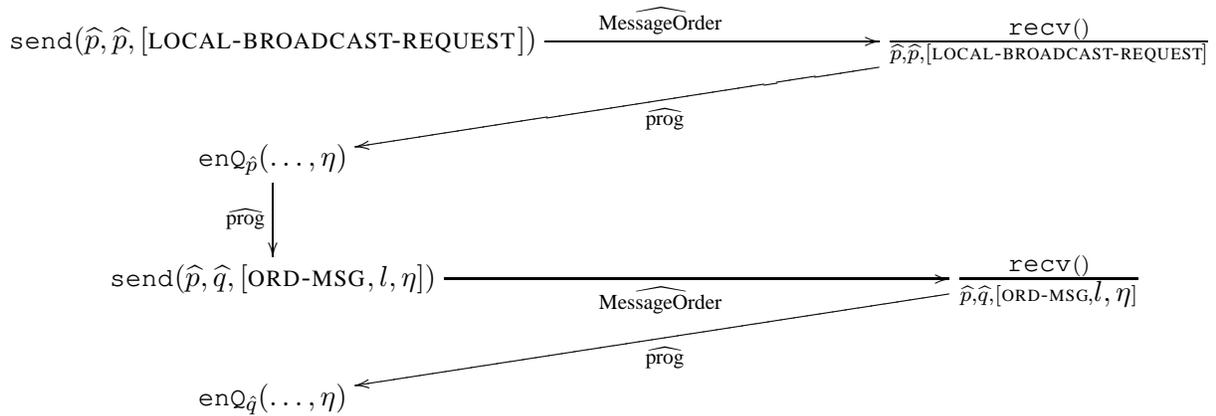
\begin{xy}
\xymatrix @C 1.75in {
\opSend{\widehat{p}, \widehat{p}, [\const{local-broadcast-request}]} \ar[r]^{\MessageOrder{}}& 
\opRecv{\widehat{p},\widehat{p},[\const{local-broadcast-request}]} \ar[ld]^{\ProgramOrder{}} & \\
\opEitherEnqueuep{\dots,\mbox{$\eta$} }{\hat{p}} \ar[d]_{\ProgramOrder{}} \\
\opSend{\widehat{p},\widehat{q},  [\const{ord-msg}, \mbox{$ l, \eta $}]} \ar[r]_{\MessageOrder{}}& 
\opRecv{\widehat{p},\widehat{q},[\const{ord-msg}, \mbox{$ l, \eta $}]}                 \ar[ld]^{\ProgramOrder{}} & \\
 \opEitherEnqueuep{\dots,\mbox{$\eta$} }{\hat{q}} \\
}
\end{xy}

\noindent
Because these are the only two ways anything is \opEitherEnqueueVoid{}ed,
(a call to \func{ProcessQueueElement} from Line \ref{li:ProcessQElementCallMe}
or from Line \ref{li:ProcessQElementCallOther})
there is a 1-1 correspondence between the set of all
\opbcastVoid\ operations by all processes in the specification level, 
and the set of all \opEitherEnqueueVoid s for each process, $\widehat{p}$ in the implementation. 
Therefore, each process eventually \opEitherEnqueueVoid s the same set of \id{queue-element}s. 
Furthermore, 
each \opEitherDequeueVoid\ operation by $\widehat{p}$ 
corresponds to exactly one \opdeliverVoid\ operation by $p$ 
(see the code for \opdeliverVoid).
Since 
only \opEitherEnqueueVoid{}ed messages can be \opEitherDequeueVoid{}ed,
it only remains to show that every 
\id{queue-element}, say $\eta$, that is \opEitherEnqueueVoid ed is eventually \opEitherDequeueVoid ed.

By Lemma \ref{timestamps-catchup-elements-lemma}, for all processes $\widehat{q}, \widehat{r} \in \widehat{P}$, 
eventually $T_{\targ{q}}[\widehat{r}] $ becomes and remains greater than or equal to $\eta.ts$. 
Therefore, eventually every \opEnqueueVoid{}ed \id{queue-element} with label $l \neq \bot$ will forever satisfy 
the timestamp part of the \func{CanExtract} predicate.
Furthermore, if some \id{priorityQ}$[l]$ contains an \id{queue-element} that satisfies 
this timestamp part of \func{CanExtract},
then the highest priority \id{queue-element} in \id{priorityQ}$[l]$ does,
because priority decreases with increasing timestamp.

We now show that eventually 
$\eta$ will also forever satisfy 
the counter part of the \func{CanExtract} or \func{CanDequeue} requirement.
Let $S_{\widehat{p}}$ be the set of \id{queue-element}s, $\gamma$, in $\widehat{p}$'s collection of queues,
such that either  
1) $\gamma$ has label $l \neq \bot$, has highest priority in \id{priorityQ[l]}, and 
satisfies the timestamp part of the \func{CanExtract} requirement,
or 
2) $\gamma$ has no label and is at the head of its \id{fifoQ}. 
We just established that this set cannot remain empty.  
Let $qe_{\widehat{p}}$ be the \id{queue-element} in $S_{\widehat{p}}$ with the least 
$(ts,\widehat{\id{source}})$ pair when it is not empty. 

Since $\widehat{\id{source}}$ sends messages in order of increasing timestamp, and channels are fifo, 
unlabeled queue-elements from \targ{\id{source}} enter \id{fifoQ}[\targ{\id{source}}] in order of increasing timestamp. 
Also, each \id{priorityQ}\ is ordered by increasing timestamp.
So every other message from $\widehat{\id{source}}$ with timestamp smaller than $ts$ must have been delivered, 
implying $\id{counter}[\id{source}]+1$ must be equal to $qe_{\widehat{p}}.counter$. 
Therefore, either $qe_{\widehat{p}}$ has a label $l \neq \bot$ and satisfies \func{CanExtract}
or $qe_{\widehat{p}}$ has label $\bot$ and satisfies \func{CanDequeue}.

Thus,   
provided only a finite number of messages are \opbcastVoid{} 
(or, in longlived computations, given a weak fairness constraint)  
every \id{queue-element} that is \opEitherEnqueueVoid ed will eventually be \opEitherDequeueVoid ed.
Therefore, $ \opbcast{m, l} \in O_C $ if and only if $\opdeliver{m, l} \in O_C|p$ for every process $p$. 
\end{proof}

\begin{lemma}
\label{timestamp-extend-program-lemma}
 $\forall p\in P: \Pred{Extends}[O|p, \orderArrow{\seqlinearp{}}, \orderArrow{\programOrder{}}]$  \\
\end{lemma}

\begin{proof}
Let $o_1, o_2 \in O|p$ such that $o_1 \orderArrow{\programOrder} o_2$.  
For $o \in \{o_1, o_2 \}$, $\tsimpl{L}{o}.\opCommon{low}{op}{}$ denotes the operation in the transformation of $o$ 
so that \lift{$\tsimpl{L}{o}.\opCommon{low}{op}{}$} $=$ $o$.

\noindent
\begin{xy}
\xymatrix{
o_1 \ar@{--}[d]_{\tsimpl{L}{}} \ar[rr]^{\programOrder{}} & & o_2 \ar@{--}[d]^{\tsimpl{L}{}}\\
\tsimpl{L}{o_1} \ar[rr]^{\ProgramOrder{}} &  & \tsimpl{L}{o_2} \\
\tsimpl{L}{o_1}.\opCommon{low}{op$_1$}{} \ar@{-(}[u] \ar[rr]^{\ProgramOrder{}} & \ar@{=>}[d]_{1.} & \tsimpl{L}{o_2}.\opCommon{low}{op$_2$}{} \ar@{-(}[u] \\
\tsimpl{L}{o_1}.\opCommon{low}{op$_1$}{}  \ar[rr]_{\seqLinearp{}} & & \tsimpl{L}{o_2}.\opCommon{low}{op$_2$}{}  \\
o_1 \ar@{--}[u]^{\func{lift}} \ar[rr]^{\seqlinearp{}} & & o_2 \ar@{--}[u]_{\func{lift}} \\
}
\end{xy}

\noindent 1. $\Pred{Extends}[O|\widehat{p}, \orderArrow{\seqLinearp}, \orderArrow{\ProgramOrder}]$ by the definition of the network model.

\end{proof}

\begin{lemma}
\label{valid-lemma}
$\forall p\in P: (O|p, \orderArrow{\seqlinearp} ) $ is a valid total order.
\end{lemma}

\begin{proof}
$\seqLinearp{}$ is valid. 
Hence, 
the modified $\seqLinearp{}$ after step 1, 
formed by removing all operations on some subsets of objects, is valid.  
After step 2, the sequence remains valid because removing some  \opSendVoid{} and \opRecvVoid{} operations maintains 
the required validity property:  a sequence of message operations is valid if it contains 
at most one \opSendVoid{} and \opRecvVoid{} of any message. 
After step 3, the subsequence of \func{Short}(\seqLinearp{}) consisting of variables and network operations is valid. 
However, the subsequence of \func{Short}(\seqLinearp{}) consisting of queue operations contains only 
\opEitherDequeuepVoid{\targ{p}} operations and is not valid. 
We now show that validity is restored in $L_p =$ \Lift{${Short}(\seqLinearp)$}.
 
The subsequence of \seqlinearp{} consisting of only variables is valid since \lift{}
is essentially an identity map for operations on variables. 
It remains to show validity for  \opbcastpVoid{p} and \opdeliverpVoid{p} operations. 
Recall that a sequence of \opbcastpVoid{p} and \opdeliverpVoid{p} operations is valid if
\begin{enumerate}
\item  
\opdeliverpVoid{p} does not precede its corresponding \opbcastpVoid{p}, and 
\item
no specific \opdeliverpVoid{p} occurs more than once. 
\end{enumerate}
The proof of Lemma \ref{bcast-deliver-lemma} showed that each update is \opdeliverVoid{}ed 
exactly once in each \seqlinearp, establishing (2).
Let  \opSend{\widehat{p}, \widehat{p}, [\const{local-broadcast-request},u,l]} and 
\opEitherDequeuep{queue_l}{[u,ts,c,\widehat{p}]}{\targ{p}} be operations in $ O_{\widehat{C}}|\widehat{p}$ .
The next diagram establishes (1). 

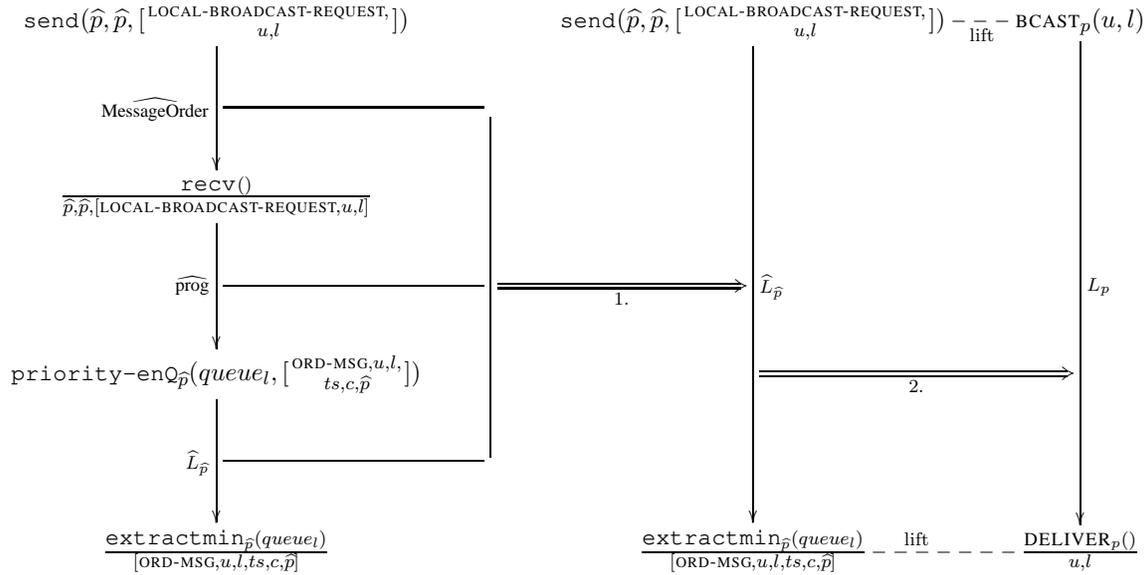
\scalebox{0.875}{
\begin{xy}
\xymatrix{
\opSend{\widehat{p}, \widehat{p}, [{\const{local-broadcast-request}, \atop u,l}]} \ar[dd]_{\MessageOrder{}} &  & \opSend{\widehat{p}, \widehat{p}, [{\const{local-broadcast-request}, \atop u,l}]} \ar[dddddd]^{\seqLinearp{}} & \opbcastp{u,l}{p} \ar[dddddd]^{\seqlinearp{}} \ar@{--}[l]^(0.3){\func{lift}}\\
\ar@{-}[r] & \ar@{-}[dddd] \\
\opRecv{\widehat{p}, \widehat{p}, [{\const{local-broadcast-request}, u, l}]} \ar[dd]_{\ProgramOrder{}} & & & \\
\ar@{-}[r] & \ar@{=>}[r]_{1.} & \\
\opEnqueuep{queue_l,[{\const{ord-msg}, u, l, \atop ts,c,\widehat{p}}]}{\targ{p}} \ar[dd]_{\seqLinearp{}} & & \ar@{=>}[r]_{2.} & \\
\ar@{-}[r] & \\
\opDequeuep{queue_l}{[\const{ord-msg},u,l, ts,c, \widehat{p}]}{\targ{p}} & & \opDequeuep{queue_l}{[\const{ord-msg},u,l, ts,c, \widehat{p}]}{\targ{p}} & \opdeliverp{u,l}{p} \ar@{--}[l]_{\func{lift}}\\
}
\end{xy}
}

\begin{enumerate}
\item{} By definition of network model $\Pred{Extends}[O|p, \orderArrow{\HappensBefore{}}, \orderArrow{\ProgramOrder} \cup \orderArrow{\MessageOrder{}} ]$ and \\
$\Pred{Extends}[O|p, \orderArrow{\seqLinearp{}}, \orderArrow{\HappensBefore{}}] $.
\item{} By construction  \orderArrow{\seqlinearp{}} from definition of \func{lift}.
\end{enumerate}

\end{proof}

\begin{lemma}
\label{timestamp-fifo-lemma}  
 $\forall p\in P: \Pred{Extends}[O|p, \orderArrow{\seqlinearp{}}, \orderArrow{\delorder_C{}}]$  
\end{lemma}

\begin{proof}

Let $o_1, o_2\in O|p$  such that $o_1 \orderArrow{\delorder{}} o_2$.
Then, by definition of \delorder{},  $o_1 = \opdeliverp{u_1,l_1}{p}$, $o_2 = \opdeliverp{u_2,l_2}{p}$ and
there is a process $q$ such that $\opbcastp{u_1,l_1}{q} \orderArrow{\programOrder} \opbcastp{u_2, l_2}{q}$.

\noindent
  \begin{xy}
    \xymatrix{
  \tsimpl{L}{\opbcastp{u_1,l_1}{q}} \ar[rr]^{\programOrder{}} & & \tsimpl{L}{\opbcastp{u_2,l_2}{q}}\\
     \opSend{\targ{q},\targ{q},[{ \const{local-broadcast-request}, \atop u_1, l_1}]} 
\ar@{-(}[u] \ar[rr]^{\ProgramOrder} \ar[d]_{\MessageOrder} & 
\ar@{=>}[d] & \opSend{\targ{q},\targ{q},[{\const{local-broadcast-request}, \atop u_2,l_2}]} \ar@{-(}[u] \ar[d]^{\MessageOrder} \\
      \func{HandleMessage}.\opRecv{\targ{q},\targ{q},[\dots ,u_1,l_1]} 
\ar[rr]_{\FifoChannel{}} \ar[d]_{\ProgramOrder} & &  
\func{HandleMessage}.\opRecv{\targ{q},\targ{q},[\dots,u_2,l_2]} \ar[d]_{\ProgramOrder}    \\
      \opWrite{\id{local-counter},x_1} \ar[d]_{\ProgramOrder}  & & \opWrite{\id{local-counter},x_2} \ar[d]_{\ProgramOrder} \\
      \opSend{\targ{q},\targ{p}, [{\const{ord-msg},l_1, \atop [u_1,\cdot ,x_1,\targ{q}]}]} \ar[d]_{\MessageOrder{}} \ar[uurr]_{\ProgramOrder} \ar[rr]^{\ProgramOrder}  & \ar@{=>}[d] & \opSend{\targ{q},\targ{p}, [{\const{ord-msg},l_2, \atop [ u_2, \cdot , x_2, \targ{q}]}]} \ar[d]_{\MessageOrder{}} 
       \\
       \opRecv{\targ{q},\targ{p}, [{\const{ord-msg},l_1, \atop [u_1,\cdot ,x_1,\targ{q}]}]} \ar[rr]_{\FifoChannel{}}\ar[d]_{\ProgramOrder{}}  & & \opRecv{\targ{q},\targ{p}, [{\const{ord-msg},l_2, \atop [ u_2, \cdot , x_2, \targ{q}]}]} \ar[d]_{\ProgramOrder{}}
       \\ 
       \opEitherEnqueue{queue_{l_1}, [u_1, \cdot ,x_1,\targ{q}]} \ar[rru]_{\ProgramOrder} & & \opEitherEnqueue{queue_{l_2}, [ u_2, \cdot , x_2, \targ{q}]} 
       \\ 
       }
  \end{xy}

Each write to \id{local-counter} increments it by 1, so $0 < x_1 < x_2 $.

To \opdeliverVoid\ each of these updates, 
\opEitherDequeuep{queue_{l_1}}{[u_1, \cdot ,x_1,\targ{q}]}{\targ{p}} 
and \opEitherDequeuep{queue_{l_2}}{[u_2, \cdot ,x_2,\targ{q}]}{\targ{p}} 
must both be in $\widehat{O}|\widehat{p}$. 
For $i \in \{1,2\}$, if label $l_i \neq \bot$, then 
$\widehat{p}$ can \opDequeuep{\id{priorityQ}[{l_i}]}{[u_i, \cdot ,x_i,\targ{q}]}{} only if
its call to \func{CanExtract(\id{priorityQ}[{l_i}])} returns \const{true},
where $[u_i, \cdot ,x_i,\targ{q}]$ is the queue-element at the head of \id{priorityQ}$[l_i]$
and $x_i$ is exactly one bigger than the value stored by $\widehat{p}$ for $counter[\targ{q}]$
(see line \ref{li:CanExtractCounterTest} of \func{CanExtract}).  
Similarly, for $i \in \{1,2\}$, if label $l_i = \bot$, then 
$\widehat{p}$ can \opFIFODequeue{\id{fifoQ}[\targ{q}]}{[u_i, \cdot ,x_i,\targ{q}]} only if
its call to \func{CanDequeue(\id{fifoQ}[\targ{q}])} returns \const{true},
where $[u_i, \cdot ,x_i,\targ{q}]$ is the queue-element at the head of \id{fifoQ}$[\targ{q}]$
and $x_i$ is exactly one bigger than the value stored by $\widehat{p}$ for $counter[\targ{q}]$
(see line \ref{li:timestampcandequeueend} of \func{CanDequeue}). 

Each process' $counter[src]$ starts at $0$ and is incremented by 
1 if and only if a message from $[src]$ is  \opEitherDequeueVoid ed 
(see line \ref{li:counterUpdate} of \opdeliverNoReturn{}).
Since $ x_1 < x_2 $, $p$ must \opEitherDequeueVoid\ $u_1$ before $u_2$:
\noindent

\begin{xy}
  \xymatrix{
  \opEitherDequeuep{queue_{l_1}}{[u_1,\cdot ,x_1,\targ{q}]}{\targ{p}}  \ar[rr]^{\ProgramOrder{}} & \ar@{=>}[d]_{1.} & \opEitherDequeuep{queue_{l_2}}{[u_2,\cdot ,x_2,\targ{q}]}{\targ{p}}\\
    \opEitherDequeuep{queue_{l_1}}{[u_1,\cdot ,x_1,\targ{q}]}{\targ{p}} \ar[rr]_(0.4){\seqLinearp{}} & \ar@{=>}[d]  & \opEitherDequeuep{queue_{l_2}}{[u_2, \cdot ,x_2,\targ{q}]}{\targ{p}} \\
     \opdeliverp{u_1, l_1}{p}  \ar@{--}[u]^{\func{lift}} \ar[rr]_{\seqlinearp{}} & & \opdeliverp{u_2, l_2}{p} \ar@{--}[u]^{\func{lift}}
  }
\end{xy}

\begin{enumerate} 
\item $\Pred{Extends}[O|\widehat{p}, \orderArrow{\seqLinearp}, \orderArrow{\ProgramOrder}]$
\end{enumerate}

\noindent This proves that  $\Pred{Extends}[O|p, \orderArrow{\seqlinearp{}}, \orderArrow{\delorder{}}]$.

\end{proof}

\begin{lemma}
\label{timestamp-agree-lemma}
 $\forall p, q\in P, l\in L: \Pred{Agree}[O|\cutdeliverslabel{l}, \orderArrow{\seqlinearp{}}, \orderArrow{\seqlinearp[q]{}}]$  \\
\end{lemma}

\begin{proof}
For each process $p$ there is a one-to-one correspondence between the set of \opdeliverNoReturn s by process $p$ 
of updates with label $l$ and the set 
of \opDequeueVoid s by process $\targ{p}$ of  queue-elements from \id{priorityQ}$[l]$. 
Let $qe_1 =[u_1, ts_1 ,\cdot  ,\targ{q}]$ and  $qe_2 =[u_2, ts_2 ,\cdot  ,\targ{r}]$  be two such queue-elements with label $l$,
where $(ts_1, q) < (ts_2,r)$.  
Then $ts_1 \leq ts_2$, so by Lemma \ref{timestamp-deliver-key-lemma}, 
for all $\widehat{p}\in\widehat{P}$,  
$ \opEnqueuep{\id{priorityQ}[l],{qe_1}}{\targ{p}} \orderArrow{\ProgramOrder{}} \opDequeuep{\id{priorityQ[l]}}{qe_2}{\targ{p}}$.
Therefore, by the definition of the priority queue (queue-elements are ordered lexicographically by (timestamp, source) pair):

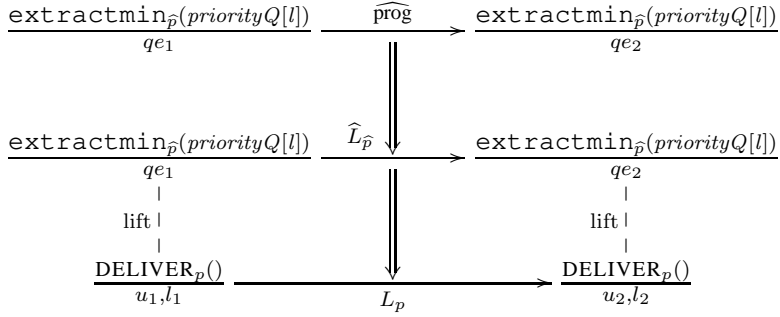
\begin{xy}
\xymatrix{
\opDequeuep{\id{priorityQ[l]}}{qe_1}{\targ{p}} \ar[rr]^{\ProgramOrder{}} &\ar@{=>}[d]& \opDequeuep{\id{priorityQ[l]}}{qe_2}{\targ{p}} \\
\opDequeuep{\id{priorityQ[l]}}{qe_1}{\targ{p}} \ar[rr]^(0.43){\seqLinearp{}} &\ar@{=>}[d]& \opDequeuep{\id{priorityQ[l]}}{qe_2}{\targ{p}} \\
    \opdeliverp{u_1, l_1}{p}  \ar@{--}[u]^{\func{lift}} \ar[rr]_{\seqlinearp{}} & & \opdeliverp{u_2, l_2}{p} \ar@{--}[u]^{\func{lift}}
}
\end{xy}

Hence, 
for each label $l$, and all processes $p,q$, the orders $([O|\cutdeliverslabel{l}, \orderArrow{\seqlinearp{}}) $  and
$(O|\cutdeliverslabel{l}, \orderArrow{\seqlinearp[q]{}})$ agree.
\end{proof}

%
%
%
%
%
%
%


\section{Summary, Open Questions and Future Work}
\label{concl.section}

This paper introduced partition consistency, a parameterized memory consistency model, 
from which other known models can be instantiated.
Four implementations of partition consistency on a message-passing network of multithreaded nodes 
were also developed and proved correct. 
All implementations are structured with a middle-level of abstraction 
which serves to modularize the implementations and simplify our proofs.
The implementations are based on Attiya and Welch's slow-write/fast-read and fast-write/slow-read 
methods \cite{attiya2004dcf}. 
Both the token-based and queue-based variants are achieved by extending Attiya and Welch's 
total order broadcast \cite{attiya2004dcf} to a partial order broadcast. 

All four implementations were proven correct using a unified framework. 
Such unified descriptions of memory consistency models at different levels of abstraction
and the associated proof techniques provide more confidence in proofs that are otherwise tedious, 
lengthy and ad hoc. 

Our proofs assume that the specification-level computations always terminate. 
Extending these proofs to long-lived computations is not involved but tedious. 
It would be useful to have a general technique to ``reduce'' the long-lived case to finite cases. 
We also suggest that the framework, the proof set-up and the diagrammatic proof descriptions used in this paper
could be used to establish the correctness of other memory consistency models for various 
multiprocess machines, or networks or languages (for example, C++).

Let us call a correct implementation \emph{exact} if every computation of the specification level  
is an interpretation of some computation of the target level. 
Our implementations in this paper are correct but not exact; 
there are computations that satisfy partition consistency that could not happen in our implementations. 
For example,
abstract memory consistency models such as P-RAM and PC-G allow a kind of cyclic causality, 
such as the computation: 

\begin{eqnarray*}
p: \ \ \opread{x}{1},  & & \opwrite{y,2}\\
q:\ \ \opread{y}{2},   & & \opwrite{x,1} 
\end{eqnarray*}

\noindent
This computation contains a cycle. 
The first process must read the value written by the second process before writing 
but the second process must read the value written by the first process before writing.
This problem can be overcome by adding a causality constraint to the memory consistency definitions.
Our implementations prohibit such cyclic computations. 
Hence, our implementations are stronger than the memory consistency models they implement
--- the specifications allow such computations, but our implementations do not. 
Though this computation may seem impossible in actual implementations, 
it could conceivably be possible if there is a prediction system in place. 
Our proof method could still be used with such a system.
As a second example, consider the computation:

\begin{center}
\begin{tabular}
{lccc}
$p :$ & \opwrite{x,1}, & \opread{y}{3}, & \opread{x}{1} \\
$q :$ & \opread{x}{1}, & \opwrite{x,2}, & \opwrite{y,3}
\end{tabular}
\end{center}

\noindent
This computation is possible in a P-RAM implementation that broadcasts to itself,
provided there is no guarantee that a process applies its own write before any other process applies that write. 
The first process broadcasts \opwrite{x,1} which is received and applied by the second process.
The second process then broadcasts \opwrite{x,2} and \opwrite{y,3}.
The first process receives \opwrite{x,2} before its own \opwrite{x,1} and so overwrites the 2 with a 1, 
even though \opwrite{x,1} ``caused'' \opwrite{x,2}.
This computation also could not happen in the implementations in this paper.  
This shows that our implementations are not exact, and 
that the simplest causality constraint added to the memory consistency definition is insufficient to make
it exact. 
Whether or not there is a simple strengthening of the partition consistency predicate to make our current 
implementations exact remains a question for future research.

A related issue to exactness is optimality. 
We believe that our implementations use minimal synchronization in the following sense. For every synchronization that is added by the transformation, there exists a program whose transformation would create computations that do not satisfy \PCName{} if that synchronization is removed. Confirming this intuition is beyond the scope of this paper.

Since the transformations in this paper are general, they are not optimal for all programs. Transforming individual programs may lead to more efficient implementations. Hence, another approach that we have not followed but is pursued by others (see Section \ref{related-work.section}) aims to determine for each program what synchronization is necessary and sufficient to preserve correctness on the target machine.

A different direction that could complement this research is the assessment of the performance gains of 
\PCName{} instantiations over sequential consistency. 
Our preliminary experiments on Westgrid's 128 node ``matrix'' cluster \cite{westgrid} are inconclusive. 
But some of these instantiations, and particularly the weak sequential consistency model, 
show a potential to outperform sequential consistency. 
This complementary study will be the subject of future work.



\section{Acknowledgments}
This research was supported by the Natural Sciences and Engineering Research Council of Canada through
discovery grant number 41900-07.
Two anonymous referees provided insightful comments that helped us improve this submission. 
\bibliographystyle{plain}
\bibliography{references}

\end{document}